\documentclass[12pt]{article}
\usepackage{graphicx}
\usepackage{nips15submit_e,times,amsmath,amssymb,graphicx}
\usepackage{algpseudocode}
\usepackage{algorithm}
\usepackage{subfigure}
\usepackage{comment}

\RequirePackage[numbers]{natbib}
\RequirePackage[colorlinks,citecolor=blue,urlcolor=blue]{hyperref}

\title{Optimal Ridge Detection using Coverage Risk}

\author{
Yen-Chi Chen\\
Department of Statistics\\
Carnegie Mellon University\\
\texttt{yenchic@andrew.cmu.edu} \\
\And
Christopher R. Genovese\\
Department of Statistics\\
Carnegie Mellon University\\
\texttt{genovese@stat.cmu.edu}\\
\And
Shirley Ho\\
Department of Physics\\
Carnegie Mellon University\\
\texttt{shirleyh@andrew.cmu.edu}\\
\And
Larry Wasserman\\
Department of Statistics\\
Carnegie Mellon University\\
\texttt{larry@stat.cmu.edu}\\
}

\newcommand\cL{\mathcal{L}}

\nipsfinalcopy

\begin{document}
\maketitle{}

\begin{abstract}
We introduce the concept of coverage risk
as an error measure for density ridge estimation.
The coverage risk
generalizes the mean integrated square error
to set estimation.
We propose two risk estimators for the coverage risk
and we show that we can select tuning parameters
by minimizing the estimated risk.
We study the rate of convergence for coverage risk
and prove consistency of the risk estimators.
We apply our method to three simulated datasets
and to cosmology data.
In all the examples, the proposed method successfully
recover the underlying density structure.
\end{abstract}

\newtheorem{thm}{Theorem}
\newtheorem{lem}[thm]{Lemma}
\newtheorem{cor}[thm]{Corollary}
\newtheorem{proposition}[thm]{Proposition}
\newenvironment{definition}[1][Definition]{\begin{trivlist}
\item[\hskip \labelsep {\bfseries #1}]}{\end{trivlist}}
\let\hat\widehat
\let\tilde\widetilde

\setlength{\parindent}{0cm}
\setlength{\parskip}{\baselineskip}

\newcommand\R{\mathbb{R}}
\newcommand\E{\mathbb{E}}
\newcommand\K{\mathbb{K}}
\newcommand\mathand{\ {\rm and}\ }
\newcommand\norm[1]{\|#1\|}
\newcommand\dest{{\sf dest}}
\newcommand\MISE{{\sf MISE}}
\newcommand\mode{{\sf Mode}}
\newcommand\ridge{{\sf Ridge}}
\newcommand\Coverage{{\sf Coverage}}
\newcommand\Haus{{\sf Haus}}
\newcommand\Loss{{\sf Loss}}
\newcommand\Risk{{\sf Risk}}
\newcommand\length{{\sf length}}
\newcommand\reach{{\sf reach}}
\newcommand\argmin{{\sf argmin}}
\newcommand\Var{{\sf Var}}
\newcommand\Cov{{\sf Cov}}
\newcommand\Tr{{\sf Tr}}

\newcommand\cI{{\cal I}}

\catcode`@=11
\newskip\beforeproofvskip
\newskip\afterproofvskip
\beforeproofvskip=\medskipamount
\afterproofvskip=\medskipamount
\def\proofsquare{\square}
\def\prooftag{Proof}
\def\proofskip{\enspace}

\def\proof{\@ifnextchar[{\@@proof}{\@proof}}  
\def\@startproof{\par\vskip\beforeproofvskip\leavevmode}
\def\@proof{\@startproof{\scshape\prooftag.}\proofskip}
\def\@@proof[#1]{\@startproof {\scshape\prooftag #1.}\proofskip}
\def\endproof{\hskip 1em $\proofsquare$\par\vskip\afterproofvskip}
\catcode`@=12

\section{Introduction}

Density ridges \citep{Eberly1996,Ozertem2011,Genovese2012a,chen2014asymptotic}
are one-dimensional curve like structures that characterize high density
regions.
Density ridges have been applied to 
computer vision \citep{bas2012local},
remote sensing \citep{miao2014method},
biomedical imaging \citep{bas2011automated},
and cosmology \citep{Chen2014GMRE,chen2015cosmic}.
Figure \ref{fig::cosmicweb} provides an example for applying density ridges
to learn the structure of our Universe.

To detect the density ridges from data,
\citep{Ozertem2011} proposed the
`Subspace Constrained Mean Shift (SCMS)' algorithm.
SCMS is a modification of usual mean shift algorithm 
\citep{fukunaga1975estimation,cheng1995mean}
to adapt to the local geometry. 
Unlike mean shift that pushes 
every mesh point to a local mode, SCMS moves the meshes 
along a projected gradient until arriving at nearby ridges.
Essentially, the SCMS algorithm detects the ridges of
the kernel density estimator (KDE).
Therefore,
the SCMS algorithm requires a pre-selected parameter $h$,
which acts as the role of smoothing bandwidth in the kernel density
estimator.

Despite the wide application of the SCMS algorithm,
the choice of $h$ remains an unsolved problem.
Similar to the density estimation problem,
a poor choice of $h$ results in 
over-smoothing or under-smoothing
for the density ridges.
See the second row of Figure \ref{fig::cosmicweb}.

In this paper, we introduce the concept of coverage risk
which is a generalization of the mean integrated expected error
from function estimation. 
We then show that one can consistently estimate the coverage risk
by using data splitting or the smoothed bootstrap.
This leads us to a data-driven selection rule for 
choosing the parameter $h$ for the SCMS algorithm.
We apply the proposed method to 
several famous datasets including the spiral dataset,
the three spirals dataset, and the NIPS dataset.
In all simulations, our selection rule allows
the SCMS algorithm to detect the underlying structure
of the data.

\subsection{Density Ridges}

Density ridges are defined as follows.
Assume $X_1,\cdots,X_n$ are independently and identically distributed from
a smooth probability density function $p$ with compact support $\K$.
The density ridges \citep{Eberly1996, genovese2014nonparametric,
chen2014asymptotic} are defined as 
$$
R=\{x\in\K: V(x)V(x)^T \nabla p(x) =0, \lambda_2(x)<0\},
$$
where $V(x) = [v_2(x),\cdots v_d(x)]$ with $v_j(x)$ being the eigenvector associated with
the ordered eigenvalue $\lambda_j(x)$ ($\lambda_1(x)\geq \cdots \geq \lambda_d(x)$) 
for Hessian matrix $H(x) = \nabla \nabla p(x)$.
That is, $R$ is the collection of points whose \emph{projected gradient} $V(x)V(x)^T \nabla p(x)=0$.
It can be shown that (under appropriate conditions),
$R$ is a collection of $1$-dimensional
smooth curves ($1$-dimensional manifolds) in $\R^d$.

The SCMS algorithm is a plug-in estimate for $R$ by using
$$
\hat{R}_n=\left\{x\in\K: \hat{V}_n(x)\hat{V}_n(x)^T \nabla \hat{p}_n(x) =0, \hat{\lambda}_2(x)<0\right\},
$$
where $\hat{p}_n(x) = \frac{1}{nh^d} \sum_{i=1}^n K\left(\frac{x-X_i}{h}\right)$
is the KDE and $\hat{V}_n$ and $\hat{\lambda}_2$ are 
the associated quantities defined by $\hat{p}_n$.
Hence, one can clearly see that the parameter $h$ in the SCMS algorithm 
plays the same role of smoothing bandwidth for the KDE.

\begin{figure}
\center
\includegraphics[width=5.4in]{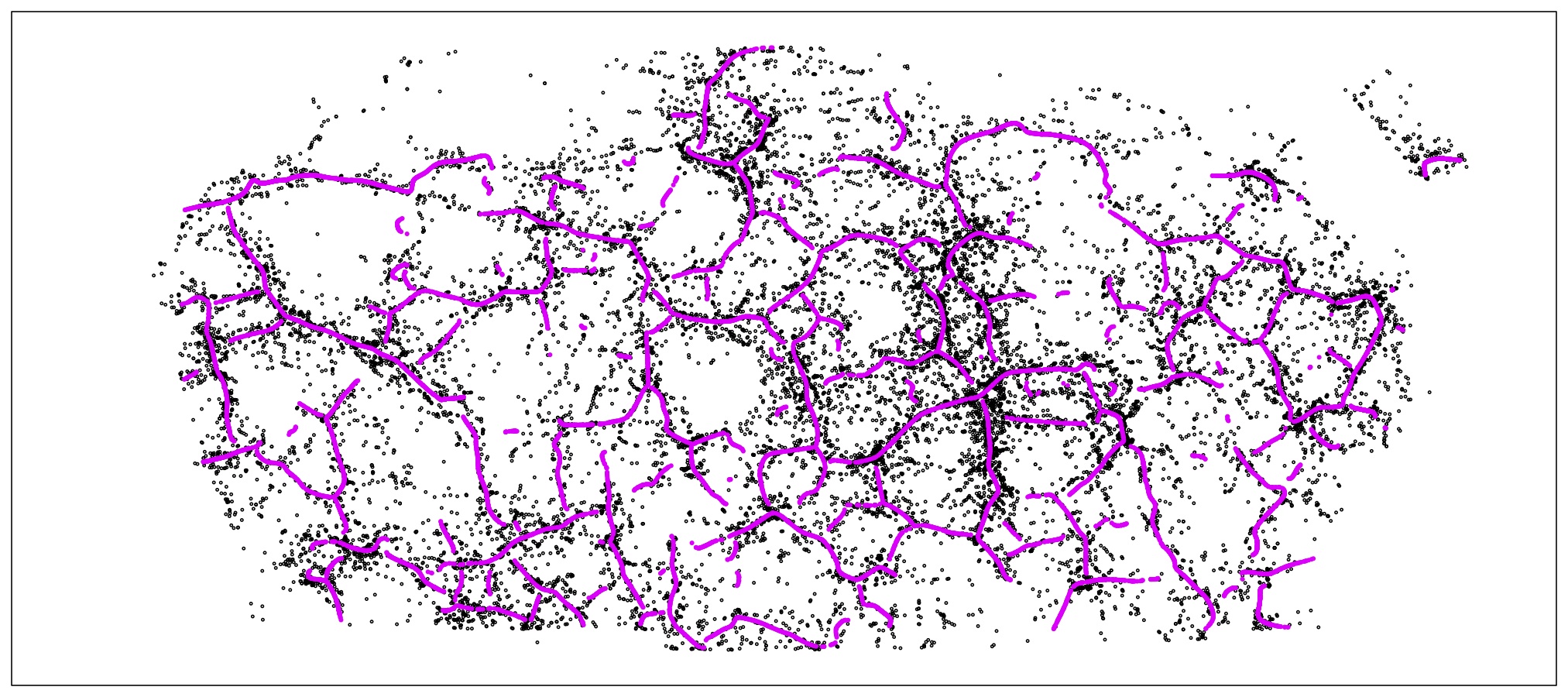}
\includegraphics[width=1.8in]{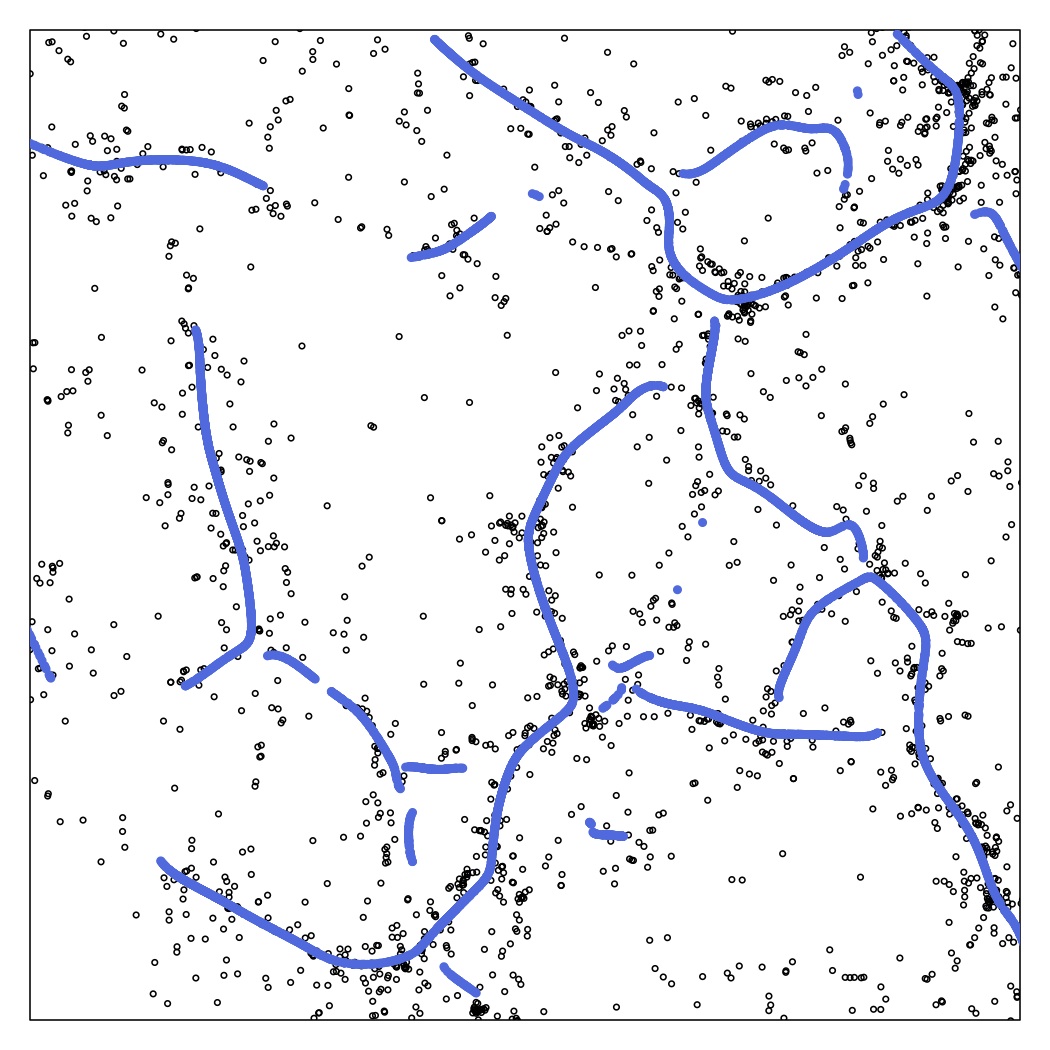}
\includegraphics[width=1.8in]{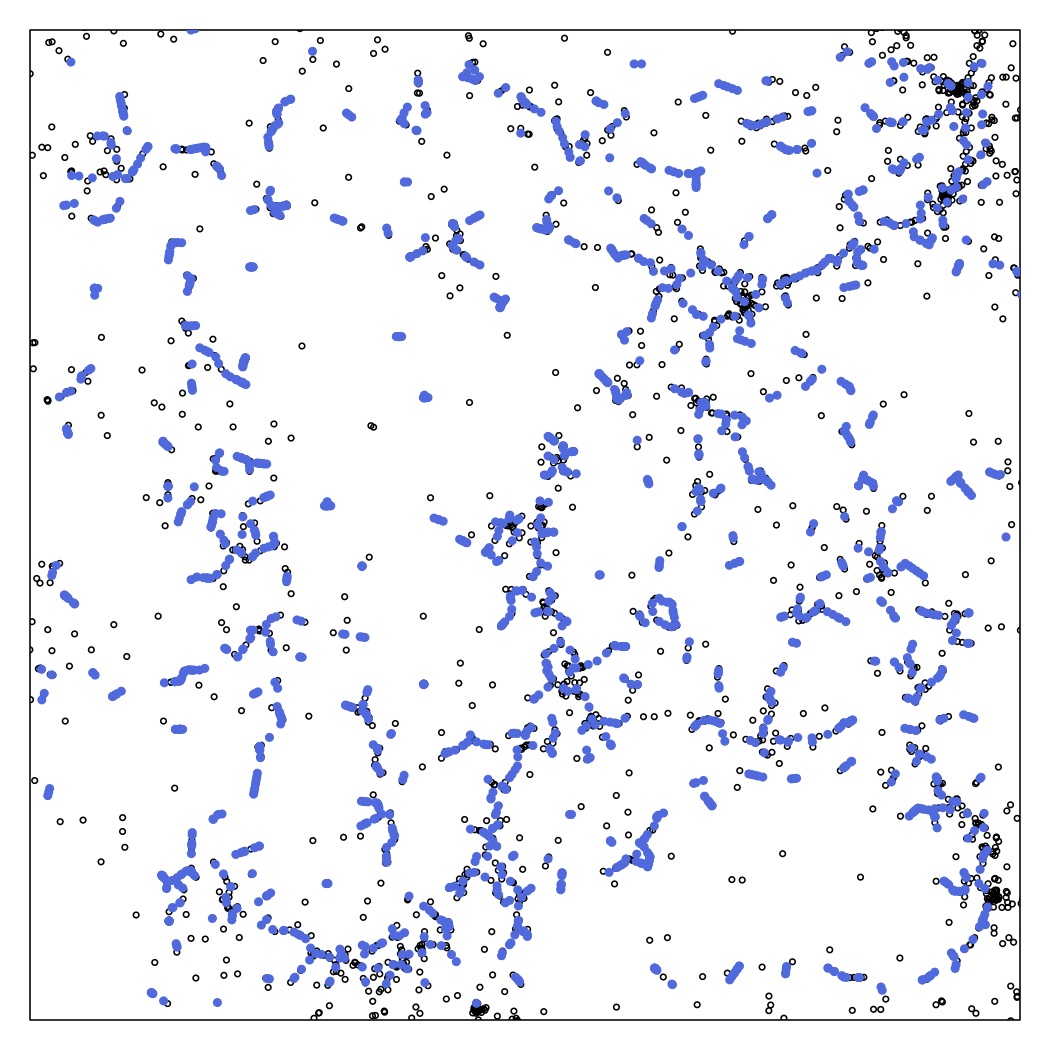}
\includegraphics[width=1.8in]{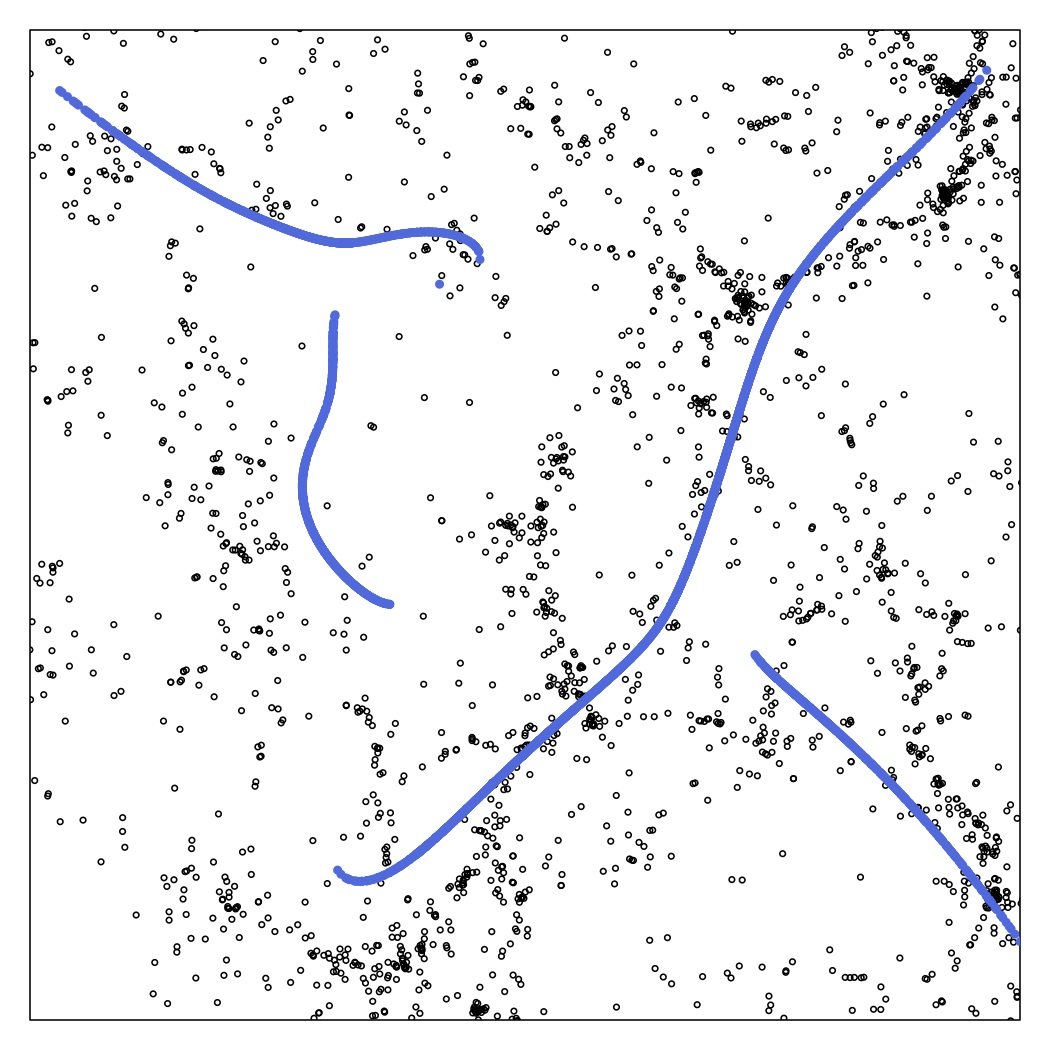}
\caption{The cosmic web.
This is a slice of the observed Universe
from the Sloan Digital Sky Survey.
We apply the density ridge method to detect filaments
\citep{chen2015cosmic}.
The top row is one example for the detected filaments.
The bottom row shows the effect of smoothing.
Bottom-Left: optimal smoothing.
Bottom-Middle: under-smoothing.
Bottom-Right: over-smoothing.
Under optimal smoothing, we detect an intricate filament network.
If we under-smooth or over-smooth the dataset,
we cannot find the structure.
}
\label{fig::cosmicweb}
\end{figure}

\section{Coverage Risk}

Before we introduce the coverage risk, we first define some geometric concepts.
Let $\mu_\ell$ be the $\ell$-dimensional Hausdorff measure \citep{evans1991measure}. 
Namely,
$\mu_1(A)$ is the length of set $A$ 
and $\mu_2(A)$ is the area of $A$.
Let $d(x, A)$ be the projection distance from point $x$ to a set $A$.
We define $U_R$ and $U_{\hat{R}_n}$ as random variables uniformly distributed over 
the true density ridges $R$
and the ridge estimator $\hat{R}_n$ respectively.
Assuming $R$ and $\hat{R}_n$ are given, we define the following two 
random variables
\begin{equation}
W_n = d(U_R, \hat{R}_n), \quad \tilde{W}_n = d(U_{\hat{R}_n}, R).
\end{equation}
Note that $U_R,U_{\hat{R}_n}$ are random variables while $R, \hat{R}_n$
are sets.
$W_n$ is the distance from a randomly selected point on $R$
to the estimator $\hat{R}_n$ and $\tilde{W}_n$
is the distance from a random point on $\hat{R}_n$ to $R$.

Let $\Haus(A,B)=\inf\{r: A\subset B\oplus r, B\subset A\oplus r\}$ be the Hausdorff distance
between $A$ and $B$ where $A\oplus r =\{x: d(x, A)\leq r\}$.
The following lemma gives some useful properties about $W_n$ and $\tilde{W}_n$.
\begin{lem}
Both random variables $W_n$ and $\tilde{W}_n$ are bounded by $\Haus(\hat{M}_n, M)$.
Namely,
\begin{equation}
0\leq W_n\leq \Haus(\hat{R}_n, R), \quad 0\leq \tilde{W}_n\leq \Haus(\hat{R}_n, R).
\end{equation}
The cumulative distribution function (CDF) for $W_n$ and $\tilde{W}_n$ are
\begin{equation}
\begin{aligned}
\mathbb{P}(W_n\leq r|\hat{R}_n) = \frac{\mu_{1}\left(R\cap (\hat{R}_n\oplus r)\right)}{\mu_{1}\left(R\right)},\quad
\mathbb{P}(\tilde{W}_n\leq r|\hat{R}_n) = 
\frac{\mu_{1}\left(\hat{R}_n\cap (R\oplus r)\right)}{\mu_{1}\left(\hat{R}_n\right)}.
\end{aligned}
\end{equation}
Thus, $\mathbb{P}(W_n\leq r|\hat{R}_n)$ is the ratio of $R$ being covered by
padding the regions around $\hat{R}_n$ at distance $r$.
\label{lem::prop}
\end{lem}
This lemma follows trivially by definition so that we omit its proof.
Lemma \ref{lem::prop} links the random variables $W_n$ and $\tilde{W}_n$
to the Hausdorff distance and the coverage for $R$ and $\hat{R}_n$.
Thus, we call them \emph{coverage} random variables.
Now we define the $\cL_1$ and $\cL_2$ \emph{coverage risk} for estimating $R$ by $\hat{R}_n$ as
\begin{equation}
\Risk_{1,n}= \frac{\mathbb{E}(W_n+\tilde{W}_n)}{2},\quad
\Risk_{2,n} = \frac{\mathbb{E}(W^2_n+\tilde{W}^2_n)}{2}.
\label{eq::risk}
\end{equation}
That is, $\Risk_{1,n}$ (and $\Risk_{2,n}$) is the expected (square) projected distance 
between $R$ and $\hat{R}_n$.
Note that the expectation in \eqref{eq::risk} applies to both $\hat{R}_n$ and $U_R$.
One can view $\Risk_{2,n}$ as a generalized mean integrated square errors (MISE)
for sets.

A nice property of $\Risk_{1,n}$ and $\Risk_{2,n}$
is that they are not sensitive to outliers of $R$ in the sense that
a small perturbation of $R$ will not change the risk too much.
On the contrary, the Hausdorff distance is very sensitive to outliers.


\subsection{Selection for Tuning Parameters Based on Risk Minimization}	\label{sec::tuning}

In this section, we will show how to choose $h$ by minimizing 
an estimate
of the risk.

We propose two risk estimators.
The first estimator is based on the \emph{smoothed bootstrap} \citep{silverman1987bootstrap}.
We sample $X^*_1,\cdots X^*_n$ from the KDE $\hat{p}_n$
and recompute the estimator $\hat{R}^*_n$.
The we estimate the risk by
\begin{equation}
\hat{\Risk}_{1,n}= \frac{\mathbb{E}(W^*_n+\tilde{W}^*_n|X_1,\cdots,X_n)}{2},\quad
\hat{\Risk}_{2,n} = \frac{\mathbb{E}(W^{*2}_n+\tilde{W}^{*2}_n|X_1,\cdots,X_n)}{2},
\end{equation}
where $W^*_n= d(U_{\hat{R}_n}, \hat{R}^*_n)$ and $\tilde{W}^*_n= d(U_{\hat{R}^*_n}, \hat{R}_n)$.

The second approach is to use data splitting.
We randomly split the data into
$X^\dagger_{11},\cdots, X^\dagger_{1m}$ and $X^\dagger_{21},\cdots, X^\dagger_{2m}$,
assuming $n$ is even and $2m=n$.
We compute the estimated manifolds by using half of the data,
which we denote as $\hat{R}^\dagger_{1,n}$ and $\hat{R}^\dagger_{2,n}$.
Then we compute 
\begin{equation}
\hat{\Risk}^\dagger_{1,n}= \frac{\mathbb{E}(W^\dagger_{1,n}+W^\dagger_{2,n}|X_1,\cdots,X_n)}{2},\quad
\hat{\Risk}^\dagger_{2,n} = \frac{\mathbb{E}(W^{\dagger2}_{1,n}+W^{\dagger2}_{2,n}|X_1,\cdots,X_n)}{2},
\end{equation}
where $W^\dagger_{1,n} = d(U_{\hat{R}^\dagger_{1,n}}, \hat{R}^\dagger_{2,n})$
and $W^\dagger_{2,n} = d(U_{\hat{R}^\dagger_{2,n}}, \hat{R}^\dagger_{1,n})$.

Having estimated the risk, we
select $h$ by
\begin{equation}
h^* = \underset{{h\leq \bar{h}_n}}{\argmin}\,\,\, \hat{\Risk}^\dagger_{1,n},
\label{eq::h_select}
\end{equation}
where $\bar{h}_n$ is an upper bound by the normal reference rule \citep{Silverman1986}
(which is known to oversmooth, so that we only select $h$ below this rule).
Moreover, one can choose $h$ by minimizing $\cL_2$ risk as well.

In \cite{einbeck2011bandwidth}, they consider
selecting the smoothing bandwidth for local principal curves by self-coverage.
This criterion is a different from ours.
The self-coverage counts data points.
The self-coverage is a monotonic increasing
function and they propose to select
the bandwidth such that the derivative is highest.
Our coverage risk yields a simple trade-off curve
and one can easily pick the optimal bandwidth by minimizing the estimated risk.


\section{Manifold Comparison by Coverage}

The concepts of coverage in previous section can be generalized 
to comparing two manifolds.
Let $M_1$ and $M_2$ be an $\ell_1$-dimensional and an $\ell_2$-dimensional
manifolds ($\ell_1$ and $\ell_2$ are not necessarily the same). 
We define the coverage random variables
\begin{equation}
W_{12} = d(U_{M_1}, M_2),\quad W_{21} = d(U_{M_2}, M_1).
\end{equation}
Then by Lemma~\ref{lem::prop}, the CDF for $W_{12}$ and $W_{21}$
contains information about how $M_1$ and $M_2$
are different from each other:
\begin{equation}
\begin{aligned}
\mathbb{P}(W_{12}\leq r) = \frac{\mu_{\ell_1}\left(M_1\cap (M_2\oplus r)\right)}{\mu_{\ell_2}\left(M_1\right)},\quad
\mathbb{P}(W_{21}\leq r) = \frac{\mu_{\ell_2}\left(M_2\cap (M_1\oplus r)\right)}{\mu_{r_2}\left(M_1\right)}.
\end{aligned}
\end{equation}
$\mathbb{P}(W_{12}\leq r)$ is the coverage on $M_1$
by padding regions with distance $r$ around $M_2$.

We call the plots of the CDF of $W_{12}$ and $W_{21}$ \emph{coverage diagrams}
since they are linked to the coverage over $M_1$ and $M_2$.
The coverage diagram allows us to study how two manifolds 
are different from each other.
When $\ell_1=\ell_2$, the coverage diagram can be used as
a similarity measure for two manifolds. 
When $\ell_1\neq \ell_2$, the coverage diagram serves as 
a measure for quality of representing high dimensional objects
by low dimensional ones.
A nice property for coverage diagram is that 
we can approximate the CDF for $W_{12}$ and $W_{21}$
by a mesh of points (or points uniformly distributed) over $M_1$ and $M_2$.
In Figure~\ref{fig::exhelix} we consider a Helix dataset whose support has dimension 
$d=3$
and we compare two curves, a spiral curve (green)
and a straight line (orange), to represent the Helix dataset.
As can be seen from the coverage diagram (right panel),
the green curve has better coverage at each distance 
(compared to the orange curve)
so that the spiral curve provides a better representation for
the Helix dataset.

\begin{figure}
\includegraphics[width=2.5in]{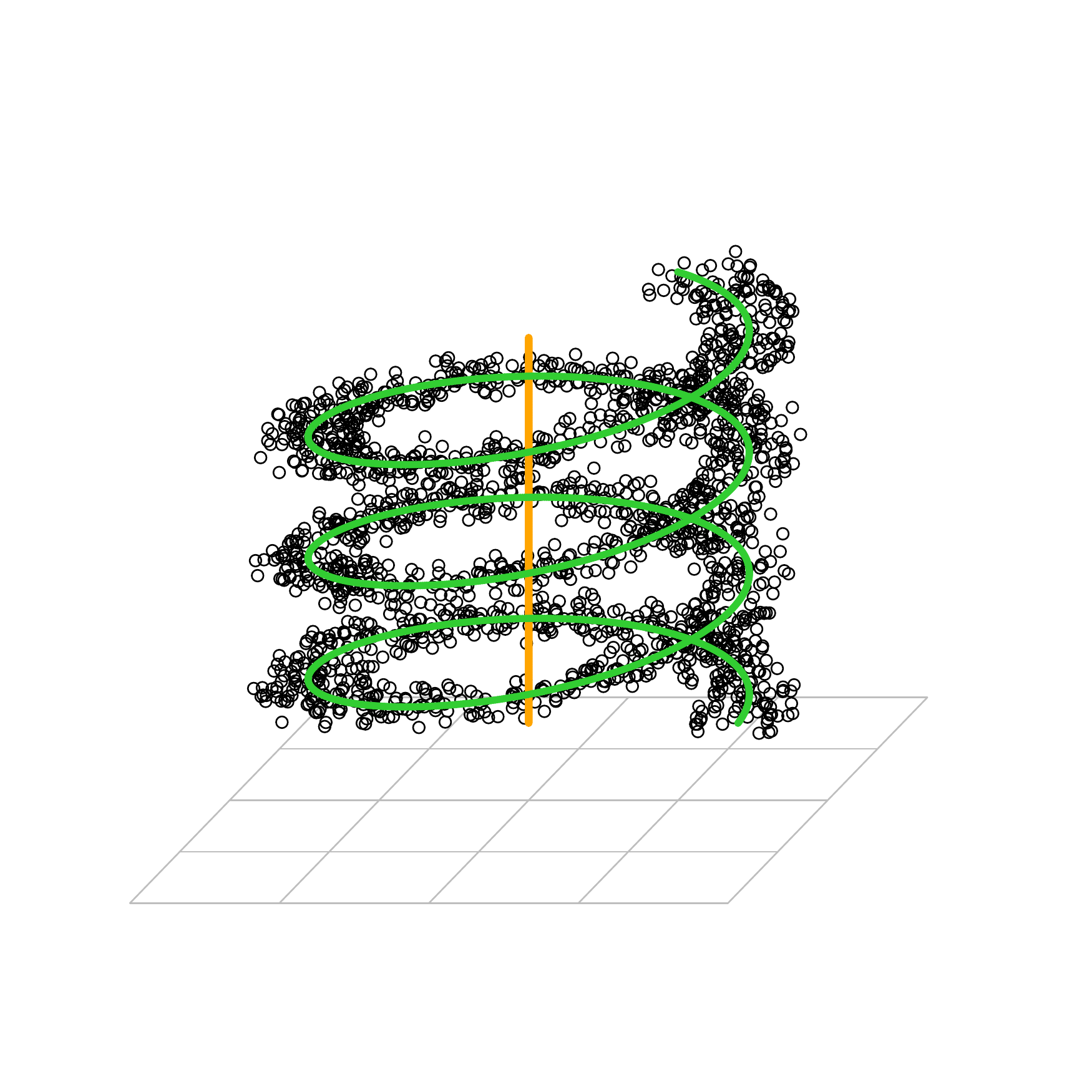}
\includegraphics[width=2.5in]{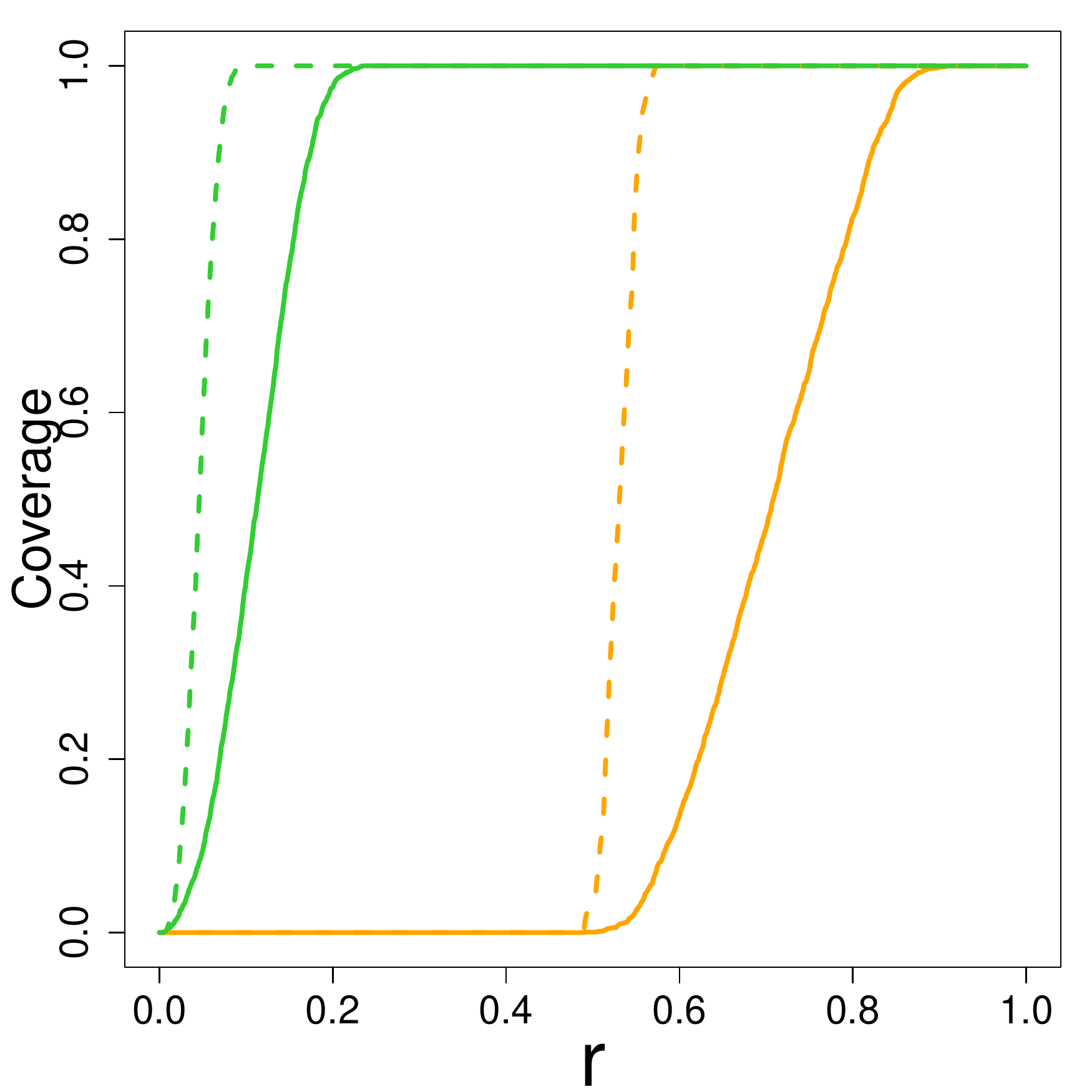}
\caption{The Helix dataset.
The original support for the Helix dataset (black dots) are a 3-dimensional regions.
We can use green spiral curves ($d=1$) to represent the regions.
Note that we also provide a bad representation using a straight line (orange).
The coverage plot reveals the quality for representation.
Left: the original data.
Right: the coverage plot for the spiral curve (green) versus a straight line (orange).
}
\label{fig::exhelix}
\end{figure}

In addition to the coverage diagram,
we can also use the following $\cL_1$ and $\cL_2$ losses as summary for
the difference:
\begin{equation}
\Loss_1(M_1, M_2)= \frac{\mathbb{E}(W_{12}+W_{21})}{2},\quad
\Loss_2 (M_1, M_2) = \frac{\mathbb{E}(W^2_{12}+W^2_{21})}{2}.
\end{equation}
The expectation is take over $U_{M_1}$ and $U_{M_2}$ and
both $M_1$ and $M_2$ here are fixed.
The risks in \eqref{eq::risk} are the expected losses:
\begin{equation}
\Risk_{1,n} = \mathbb{E}\left(\Loss_1(\hat{M}_n, M)\right), \quad
\Risk_{2,n} = \mathbb{E}\left(\Loss_2(\hat{M}_n, M)\right).
\end{equation}



\section{Theoretical Analysis}

In this section, we analyze the asymptotic behavior for 
the coverage risk and prove the consistency 
for estimating the coverage risk by the proposed method.
In particular, we derive the asymptotic properties for the density ridges.
We only focus on $\cL_2$ risk 
since by Jensen's inequality, the $\cL_2$ risk
can be bounded by the $\cL_1$ risk.

Before we state our assumption, we first define the orientation
of density ridges.
Recall that the density ridge $R$ is a collection of one dimensional
curves.
Thus, for each point $x\in R$, we can associate a unit vector $e(x)$
that represent the orientation of $R$ at $x$.
The explicit formula for $e(x)$ can be found in Lemma 1 of \cite{chen2014asymptotic}.


\textbf{Assumptions.}
\begin{itemize}
\item[\bf (R)] There exist $\beta_0, \beta_1, \beta_2,\delta_R>0$ such that
for all $x\in R\oplus \delta_R$,
\begin{equation}
\begin{aligned}
\lambda_2(x)\leq -\beta_1, \quad \lambda_1(x)\geq \beta_0-\beta_1,\quad
\norm{\nabla p(x)} \norm{ p^{(3)}(x)}_{\max}\leq \beta_0 (\beta_1-\beta_2),
\end{aligned}
\label{eq::A::eigen}
\end{equation}
where $\norm{ p^{(3)}(x)}_{\max}$ is the element wise norm to the third derivative.
And for each $x\in R$, $|e(x)^T \nabla p(x)|\geq \frac{\lambda_1(x)}{\lambda_1(x)-\lambda_2(x)}$.

\item[\bf (K1)] The kernel function $K\in\mathbf{BC}^3$ and is symmetric, non-negative and 
$$\int x^2K^{(\alpha)}(x)dx<\infty,\qquad \int \left(K^{(\alpha)}(x)\right)^2dx<\infty
$$ 
for all $\alpha=0,1,2,3$.
\item[\bf (K2)] The kernel function $K$ and its partial derivative satisfies condition $K_1$ in \cite{Gine2002}. Specifically, let 
\begin{equation}
\begin{aligned}
\mathcal{K} &=\left\{y\mapsto K^{(\alpha)}\left(\frac{x-y}{h}\right): x\in\mathbb{R}^d, h>0, |\alpha|=0,1,2\right\}\\
\end{aligned}
\end{equation}
We require that $\mathcal{K}$ satisfies
\begin{align}
\underset{P}{\sup} N\left(\mathcal{K}, L_2(P), \epsilon\norm{F}_{L_2(P)}\right)\leq \left(\frac{A}{\epsilon}\right)^v
\label{eq::VC}
\end{align}
for some positive number $A,v$, where $N(T,d,\epsilon)$ denotes the $\epsilon$-covering number of the metric space $(T,d)$ and $F$ is the envelope function of $\mathcal{K}$ and the supreme is taken over the whole $\mathbb{R}^d$. The $A$ and $v$ are usually called the VC characteristics of $\mathcal{K}$.
The norm $\norm{F}_{L_2(P)} = \sup_{P} \int|F(x)|^2dP(x)$.

\end{itemize}
Assumption (R) appears in \cite{chen2014asymptotic} and is very mild.
The first two inequality in \eqref{eq::A::eigen} are just the bound on eigenvalues.
The last inequality requires the density around
ridges to be smooth.
The latter part of (R) requires the direction of ridges to be similar to the gradient
direction.
Assumption (K1) is the common condition for kernel density estimator
see e.g. \cite{wasserman2006all} and \cite{scott2009multivariate}.
Assumption (K2) is to regularize the classes of kernel functions
that is widely assumed \citep{Einmahl2005,genovese2014nonparametric,chen2014nonparametric};
any bounded kernel function with compact support satisfies this condition.
Both (K1) and (K2) hold for the Gaussian kernel.

Under the above condition, we derive the rate of convergence
for the $\cL_2$ risk.
\begin{thm}
Let $\Risk_{2,n}$ be the $\cL_2$ coverage
risk for estimating the density ridges and level sets.
Assume (K1--2) and (R) and $p$ is at least four times bounded differentiable.
Then as $n\rightarrow \infty$, $h\rightarrow 0$ and $\frac{\log n}{nh^{d+6}}\rightarrow 0$
\begin{align*}
\Risk_{2,n} =
B_R^2 h^4 + \frac{\sigma_R^2}{nh^{d+2}} +o(h^4) + o\left(\frac{1}{nh^{d+2}}\right),
\end{align*}
for some $B_R$ and $\sigma_R^2$ that depends only on the density $p$ 
and the kernel function $K$.
\label{thm::rate}
\end{thm}
The rate in Theorem~\ref{thm::rate}
shows a bias-variance decomposition.
The first term involving $h^4$ is the bias term
while the latter term is the variance part.
Thanks to the Jensen's inequality, the rate of convergence for $\cL_1$ risk
is the square root of the rate Theorem~\ref{thm::rate}.
Note that we require the smoothing parameter $h$ to decay slowly to $0$
by $\frac{\log n}{nh^{d+6}}\rightarrow 0$.
This constraint comes from the uniform bound
for estimating third derivatives for $p$.
We need this constraint since we need the smoothness
for estimated ridges to converge to the smoothness for the true ridges.
Similar result for density level set appears in 
\cite{Cadre2006,mason2009asymptotic}.

By Lemma~\ref{lem::prop}, we can upper bound the $\cL_2$ risk
by expected square of the Hausdorff distance which gives the rate
\begin{equation}
\begin{aligned}
\Risk_{2,n} \leq\mathbb{E}\left( \Haus^2 (\hat{R}_n, R)\right) &= 
O(h^4) +O\left(\frac{\log n}{nh^{d+2}}\right)
\end{aligned}
\end{equation}
The rate under Hausdorff distance for density ridges can be found in 
\cite{chen2014asymptotic} and the rate for density ridges appears in \cite{Cuevas2006}.
The rate induced by Theorem~\ref{thm::rate}
agrees with the bound from the Hausdorff distance and has a slightly better
rate for variance (without a log-n factor).
This phenomena is similar to the MISE and $\cL_{\infty}$ error
for nonparametric estimation for functions.
The MISE converges slightly faster (by a log-n factor) 
than square to the $\cL_{\infty}$ error.


Now we prove the consistency of the risk estimators.
In particular, we prove the consistency for the smoothed
bootstrap. The case of data splitting can be proved in the similar way.
\begin{thm}
Let $\Risk_{2,n}$ be the $\cL_2$ coverage
risk for estimating the density ridges and level sets.
Let $\hat{\Risk}_{2,n}$ be the corresponding
risk estimator by the smoothed bootstrap.
Assume (K1--2) and (R) and $p$ is at least four times bounded differentiable.
Then as $n\rightarrow \infty$, $h\rightarrow 0$ and $\frac{\log n}{nh^{d+6}}\rightarrow 0$,
\begin{align*}
\frac{\hat{\Risk}_{2,n}-\Risk_{2,n}}{\Risk_{2,n}} \overset{P}{\rightarrow}0.
\end{align*}
\label{thm::estimate}
\end{thm}
Theorem~\ref{thm::estimate} proves the consistency
for risk estimation using the smoothed bootstrap.
This also leads to the consistency for
data splitting.



\section{Applications}

\subsection{Simulation Data}

\begin{figure}
\includegraphics[width=1.3in]{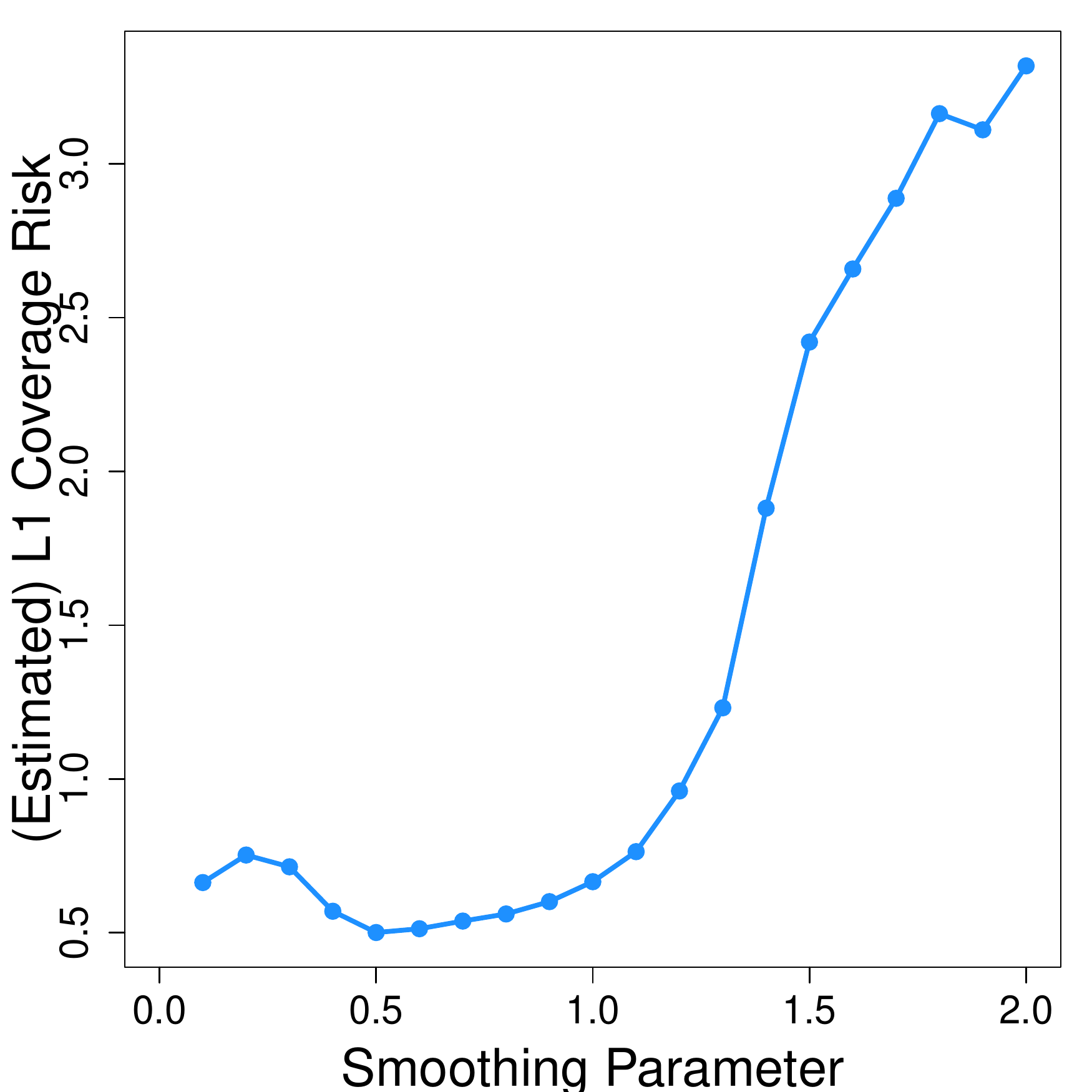}
\includegraphics[width=1.3in]{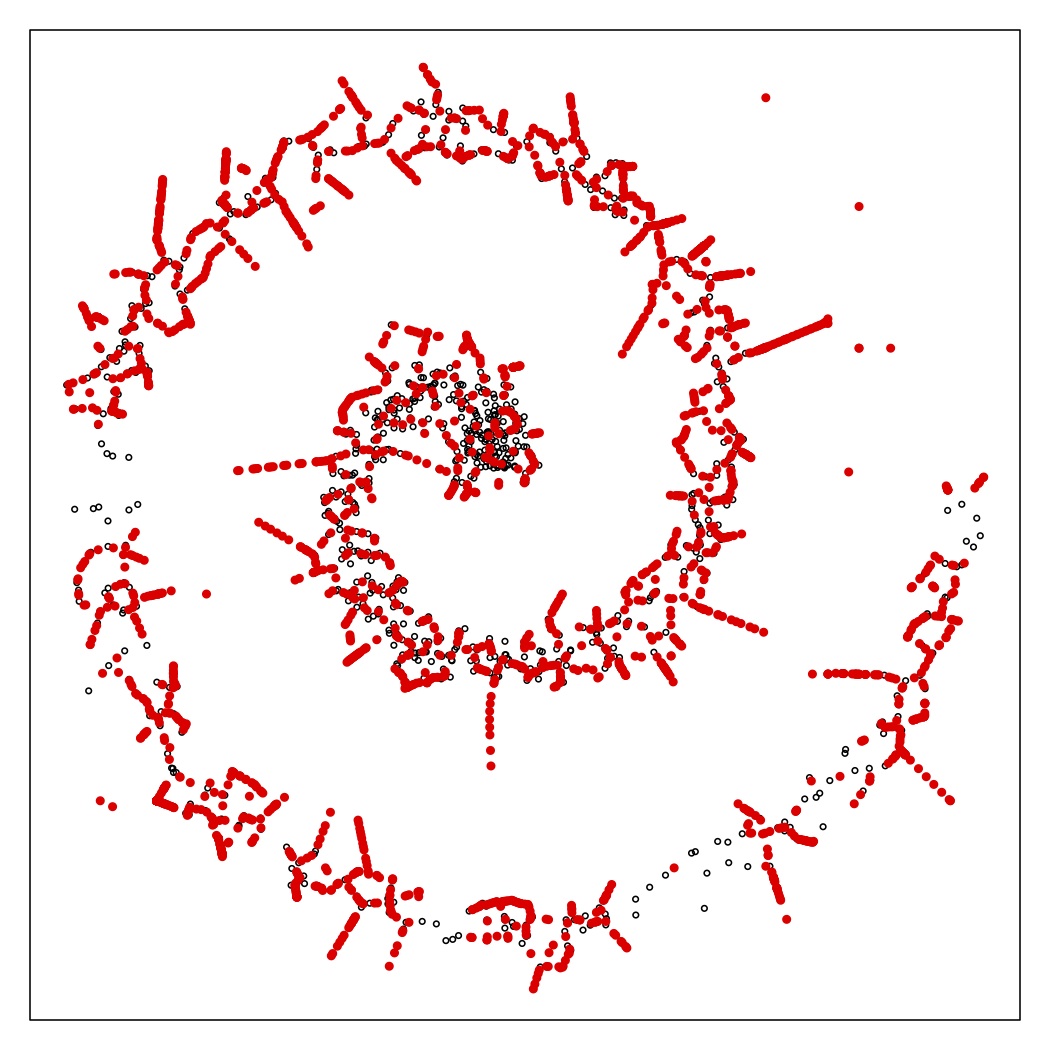}
\includegraphics[width=1.3in]{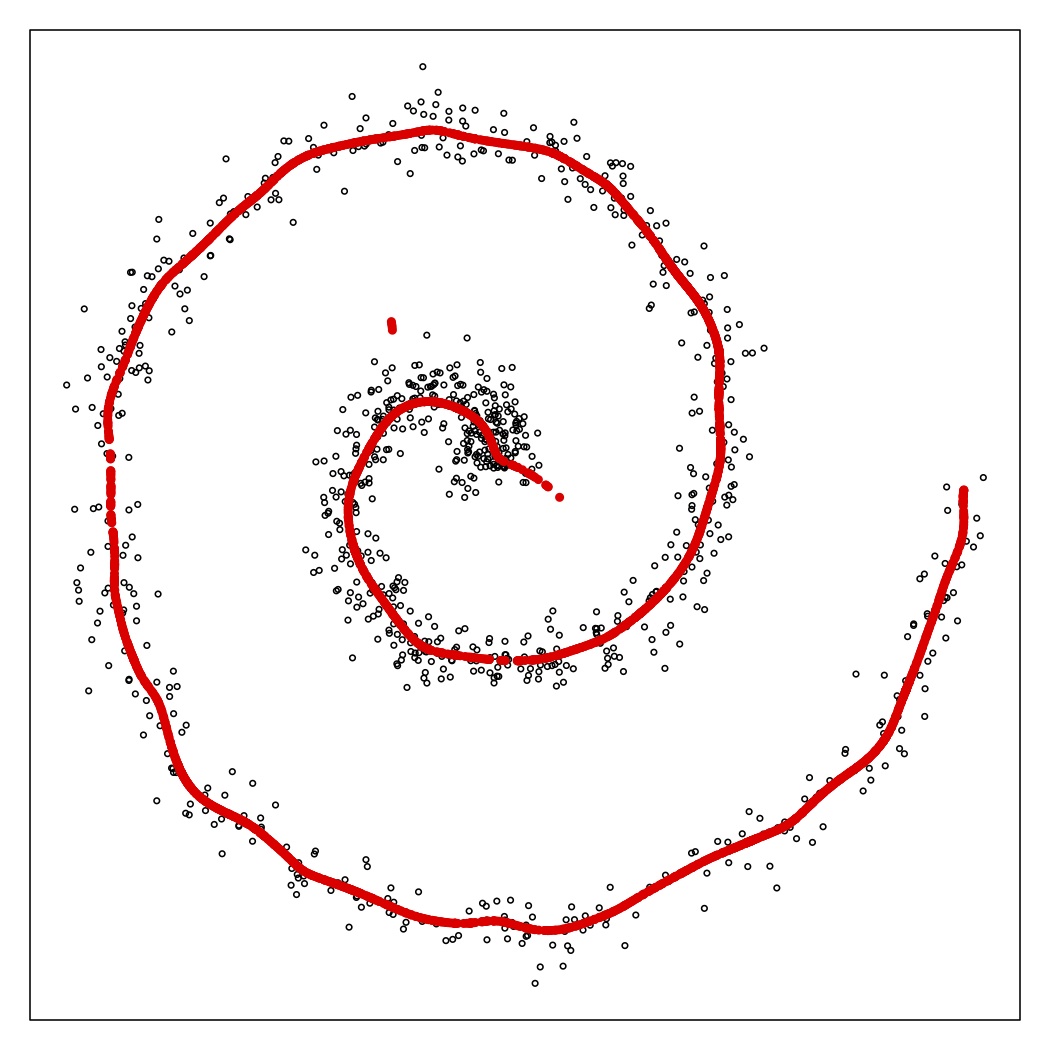}
\includegraphics[width=1.3in]{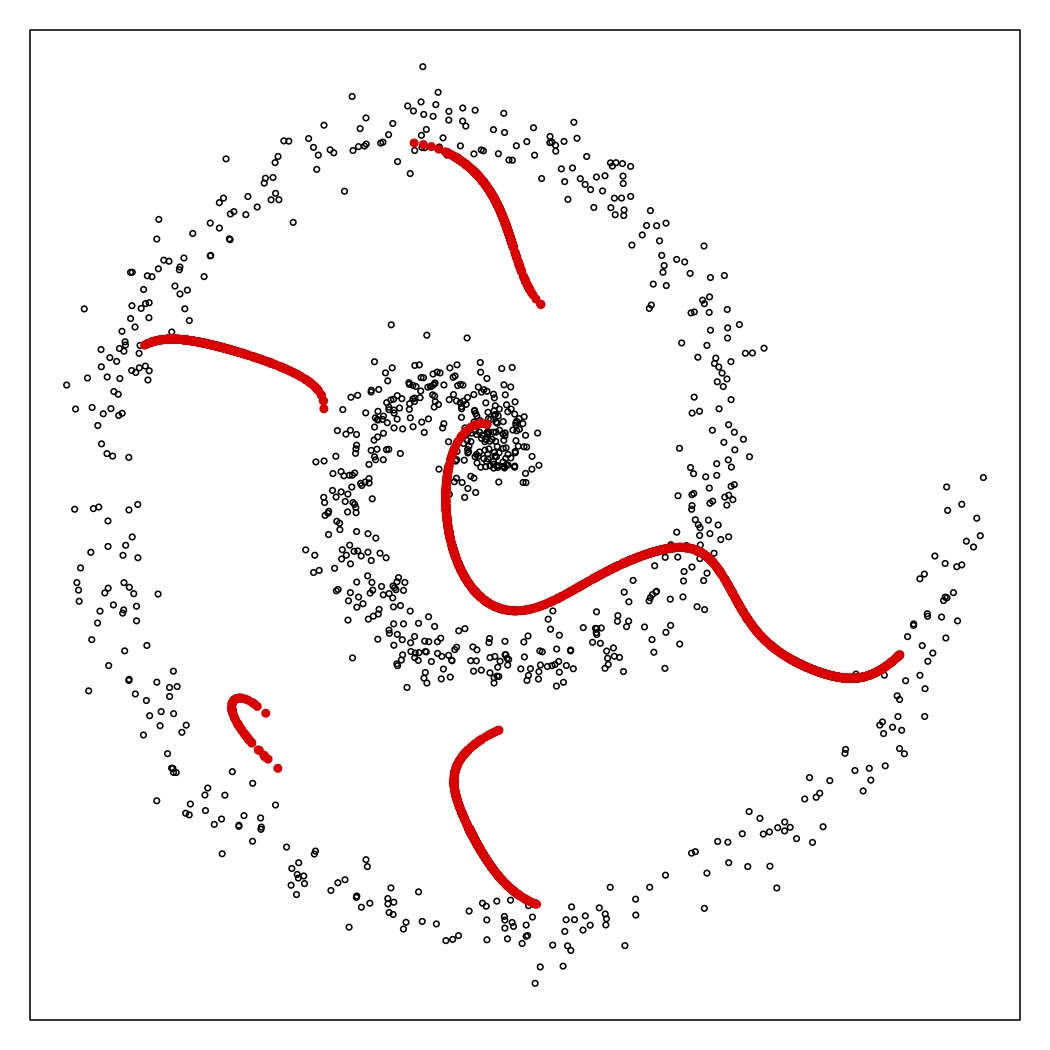}\\
\includegraphics[width=1.3in]{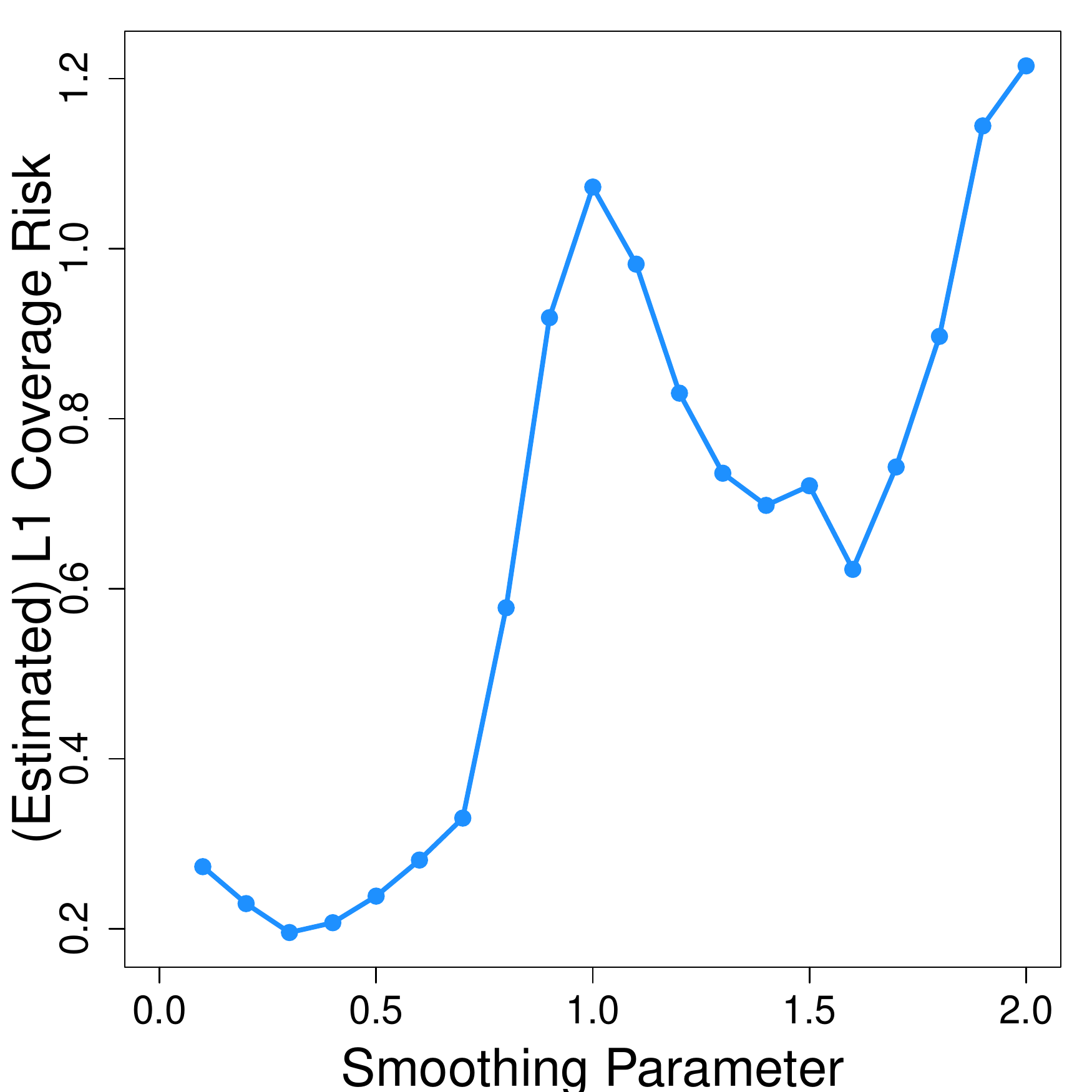}
\includegraphics[width=1.3in]{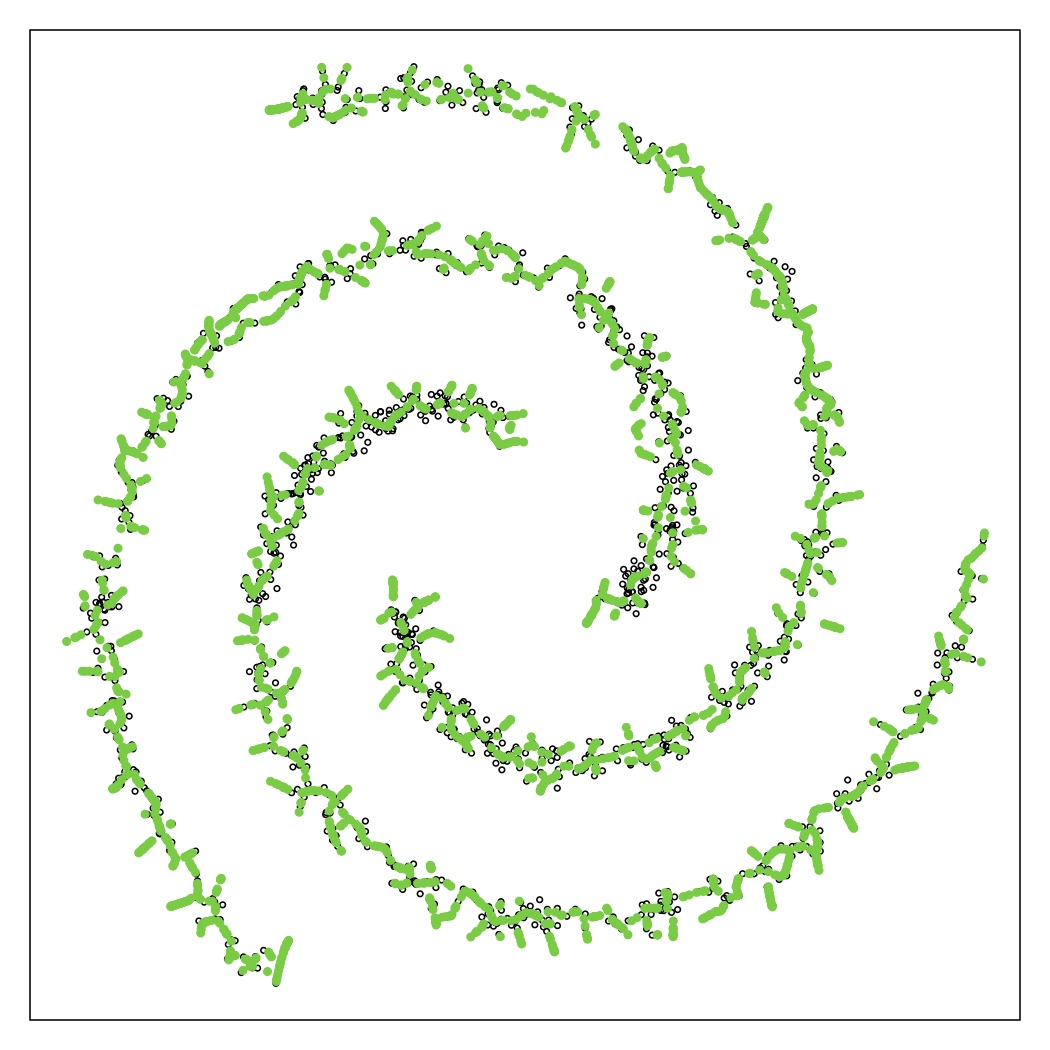}
\includegraphics[width=1.3in]{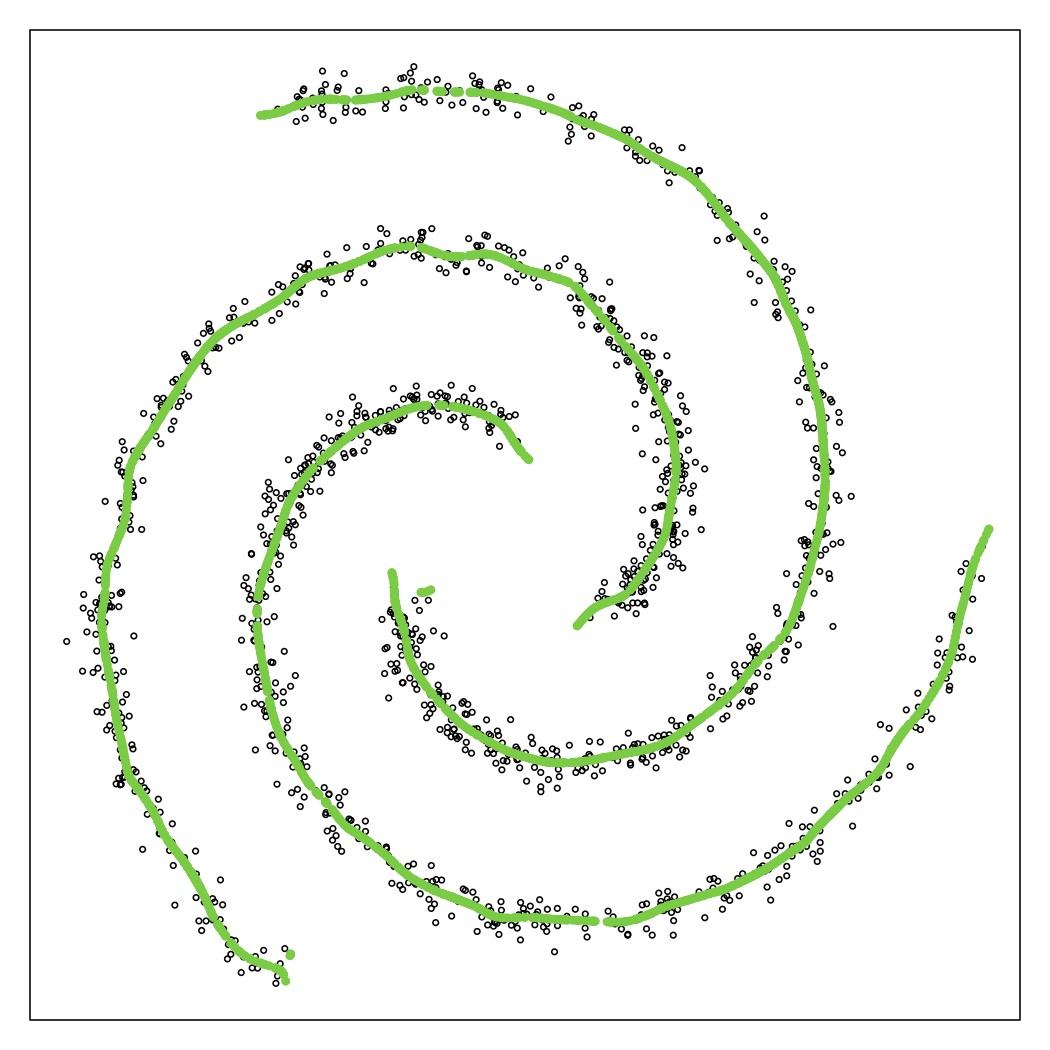}
\includegraphics[width=1.3in]{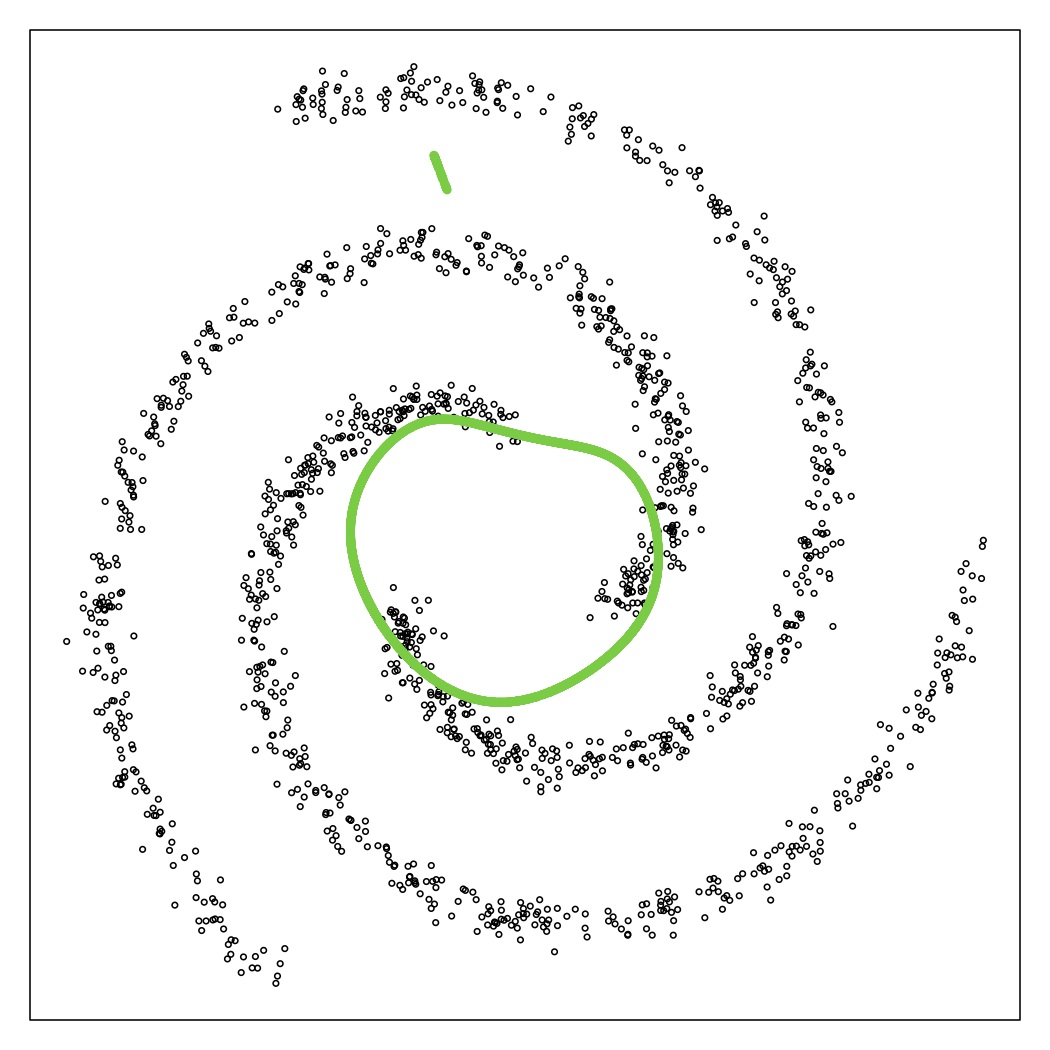}\\
\includegraphics[width=1.3in]{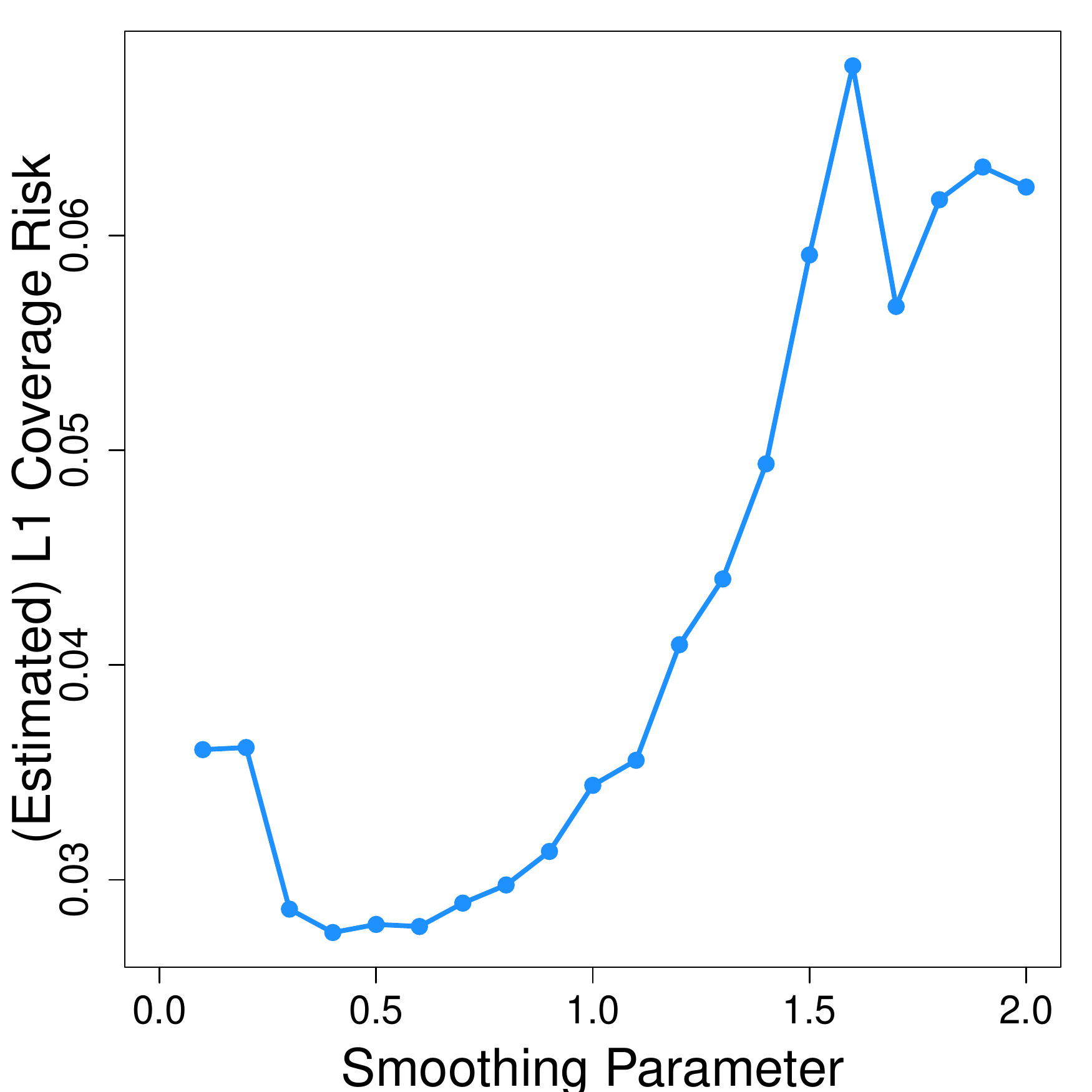}
\includegraphics[width=1.3in]{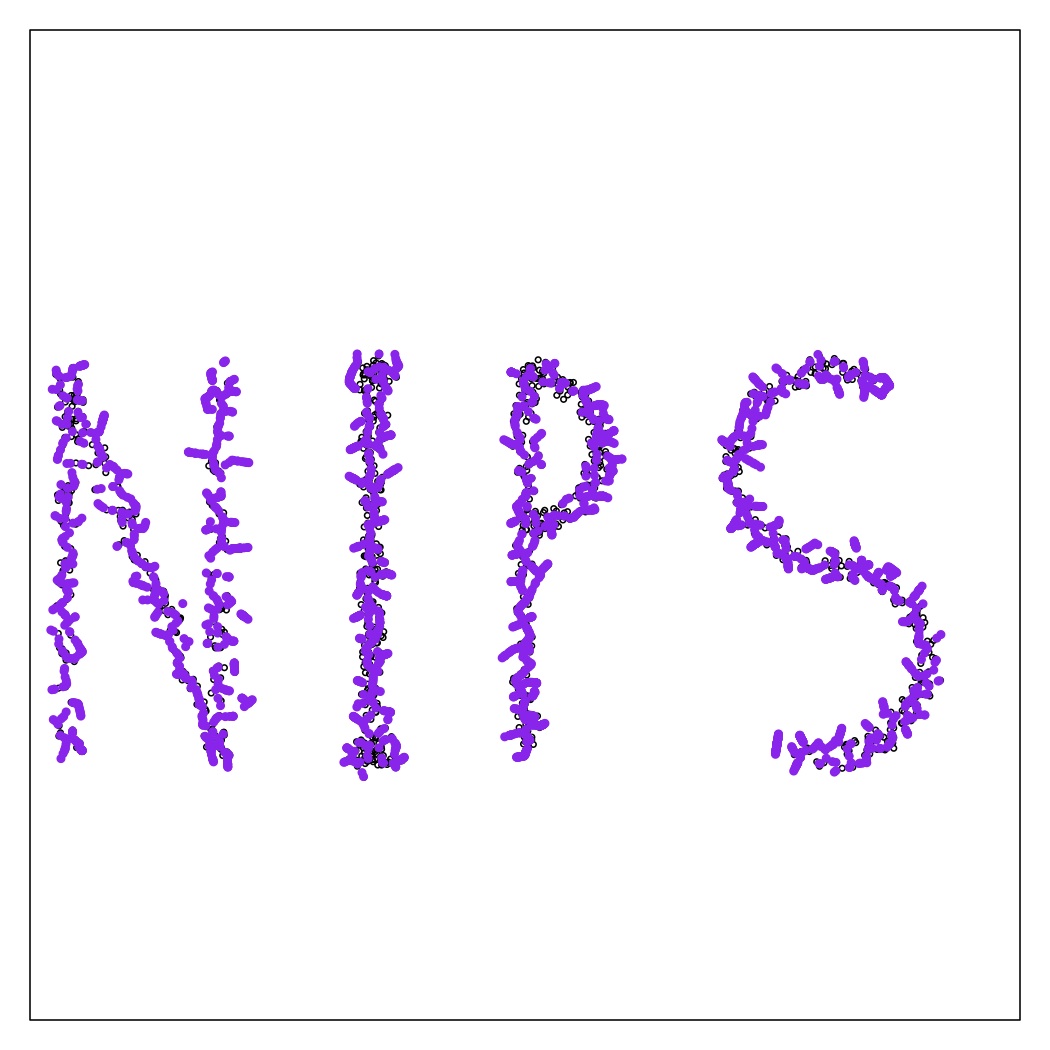}
\includegraphics[width=1.3in]{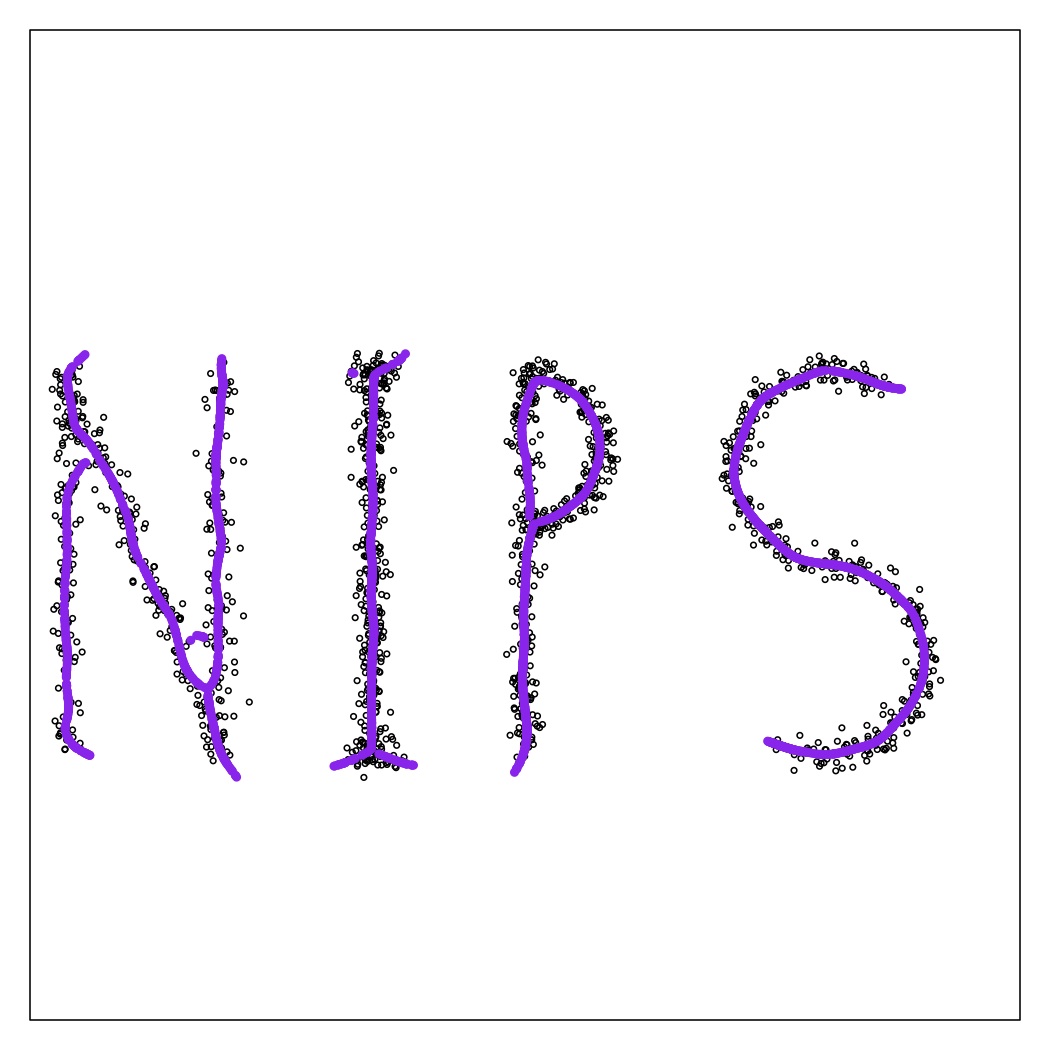}
\includegraphics[width=1.3in]{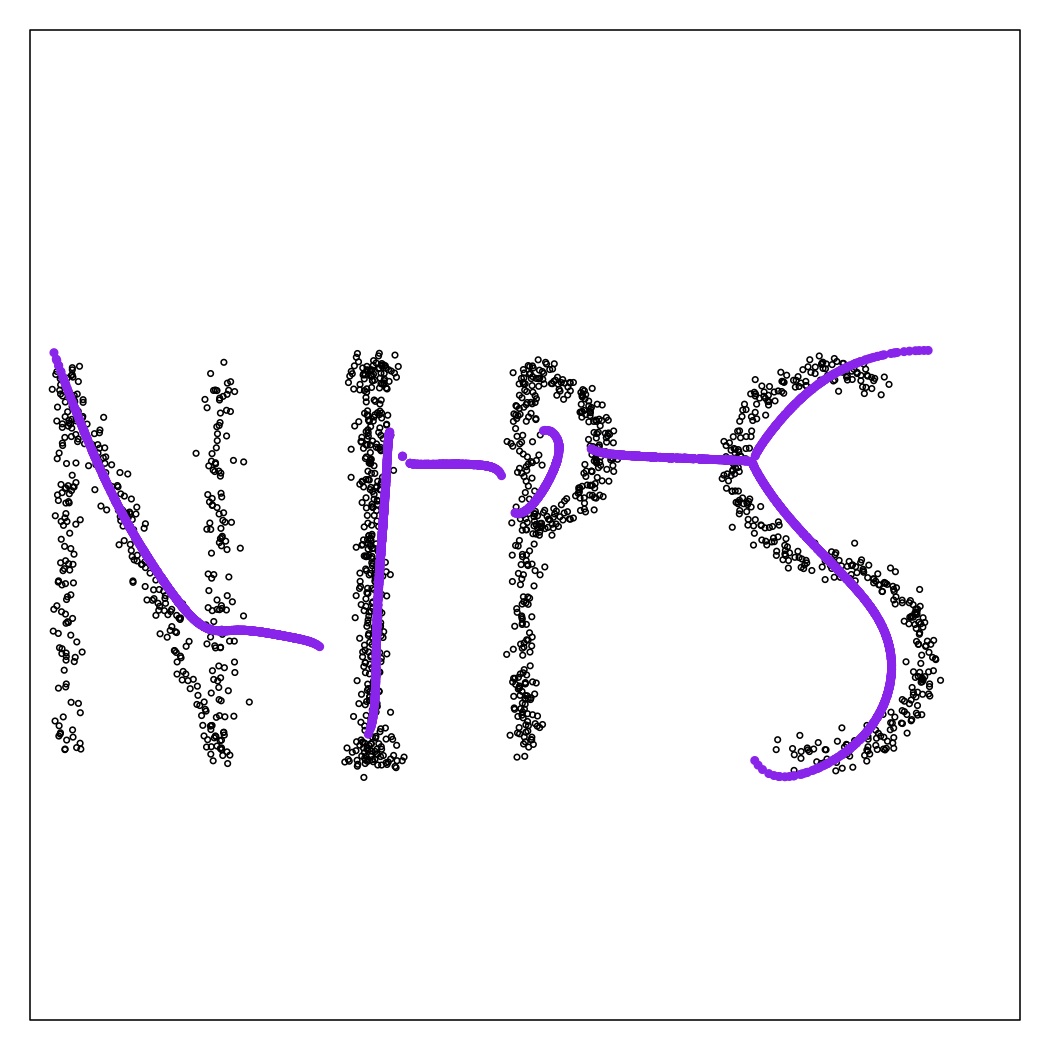}
\caption{
Three different simulation datasets.
Top row: the spiral dataset.
Middle row: the three spirals dataset.
Bottom row: NIPS character dataset.
For each row, the leftmost panel shows 
the estimated $\cL_1$ coverage risk
using data splitting.
Then the rest three panels, are the result using different smoothing parameters.
From left to right, we show the result for
under-smoothing, optimal smoothing (using the coverage risk), and over-smoothing.
}
\label{fig::exs1}
\end{figure}


We now apply the data splitting technique \eqref{eq::h_select} to choose the smoothing
bandwidth for density ridge estimation.
The density ridge estimation can be done by the subspace constrain
mean shift algorithm \citep{Ozertem2011}.
We consider three famous datasets: the spiral dataset, the three spirals dataset
and a `NIPS' dataset.


Figure~\ref{fig::exs1} shows the result for the three simulation datasets.
The top row is the spiral dataset; the middle row is the three spirals dataset;
the bottom row is the NIPS character dataset.
For each row, from left to right the first panel is 
the estimated $\cL_1$ risk by using data splitting.
The second to fourth panels are under-smoothing, optimal smoothing, 
and over-smoothing.
Note that we also remove the ridges whose density is below $0.05\times \max_x\hat{p}_n(x)$
since they behave like random noise.
As can be seen easily, the optimal bandwidth allows
the density ridges to capture the underlying structures
in every dataset.
On the contrary, the under-smoothing and the over-smoothing does not 
capture the structure and have a higher risk.

\subsection{Cosmic Web}

Now we apply our technique to the Sloan Digital Sky Survey, a huge
dataset that contains millions of galaxies.
In our data, each point is an observed galaxy with three features:
\begin{itemize}
\item z: the redshift, which is the distance from the galaxy to Earth.
\item RA: the right ascension, which is the longitude of the Universe.
\item dec: the declination, which is the latitude of the Universe.
\end{itemize}
These three features $(z,RA,dec)$ uniquely determine the location of a given galaxy.

To demonstrate the effectiveness of our method,
we select a 2-D slice of our Universe at redshift $z=0.050-0.055$
with $(RA, dec) \in [200, 240] \times [0,40]$.
Since the redshift difference is very tiny, we ignore the redshift value
of the galaxies within this region and treat them as a 2-D data points.
Thus, we only use $RA$ and $dec$.
Then we apply the SCMS algorithm 
(version of \citep{chen2015cosmic}) 
with data splitting method introduced in section \ref{sec::tuning}
to select the smoothing parameter $h$.
The result is given in Figure~\ref{fig::cosmic2}.
The left panel provides the estimated coverage risk at different smoothing bandwidth.
The rest panels give the result for under-smoothing (second panel), 
optimal smoothing (third panel) and over-smoothing (right most panel).
In the third panel of Figure~\ref{fig::cosmic2},
we see that the SCMS algorithm detects the filament structure
in the data.

\begin{figure}
\includegraphics[width=1.3in]{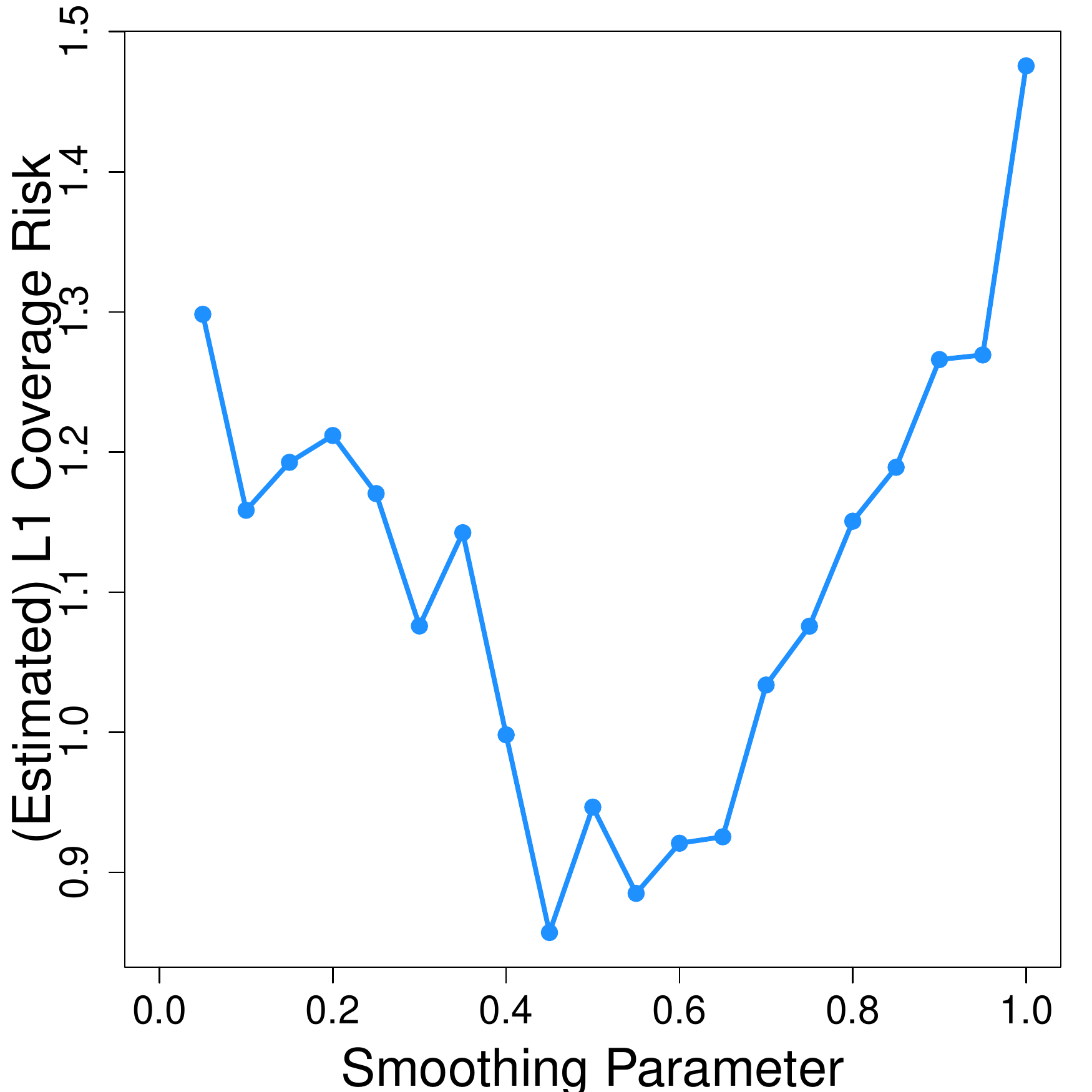}
\includegraphics[width=1.3in]{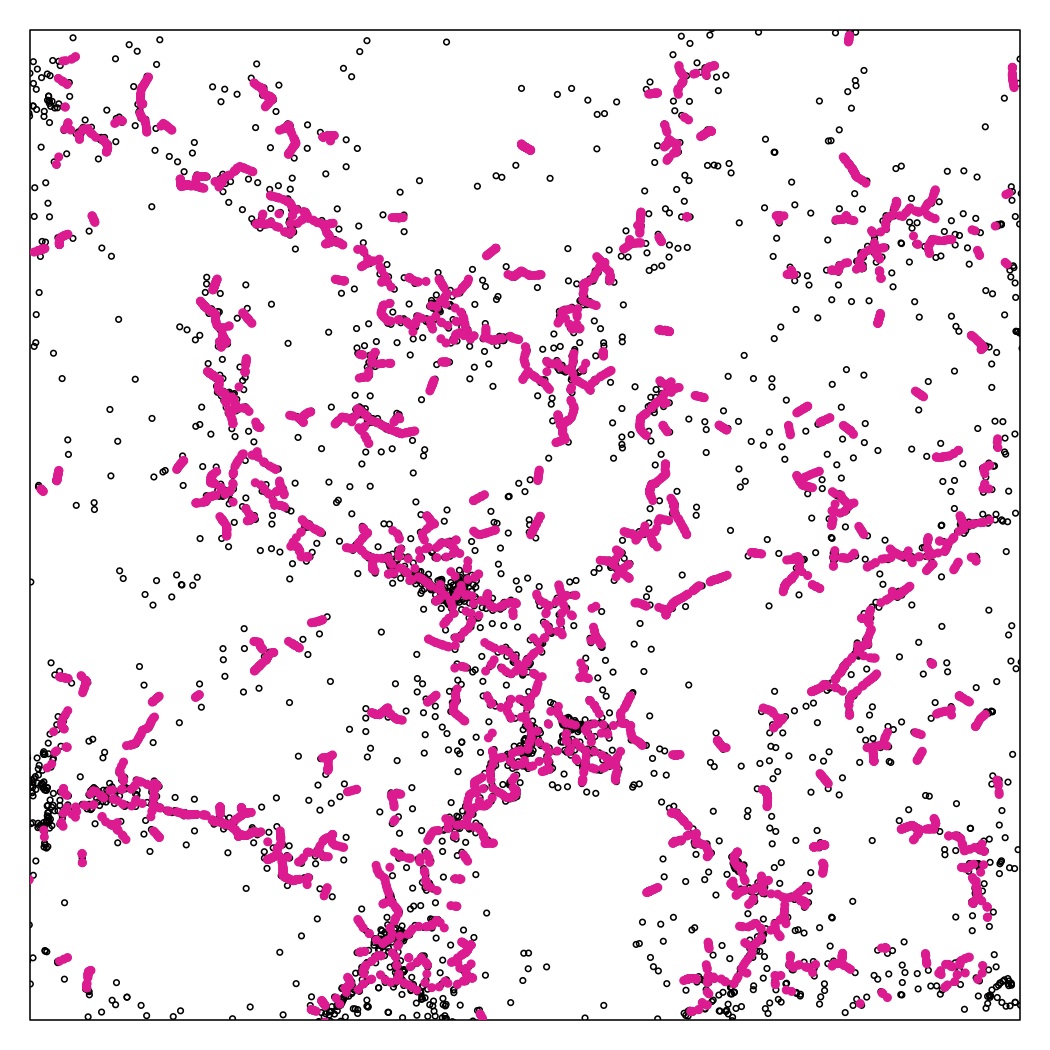}
\includegraphics[width=1.3in]{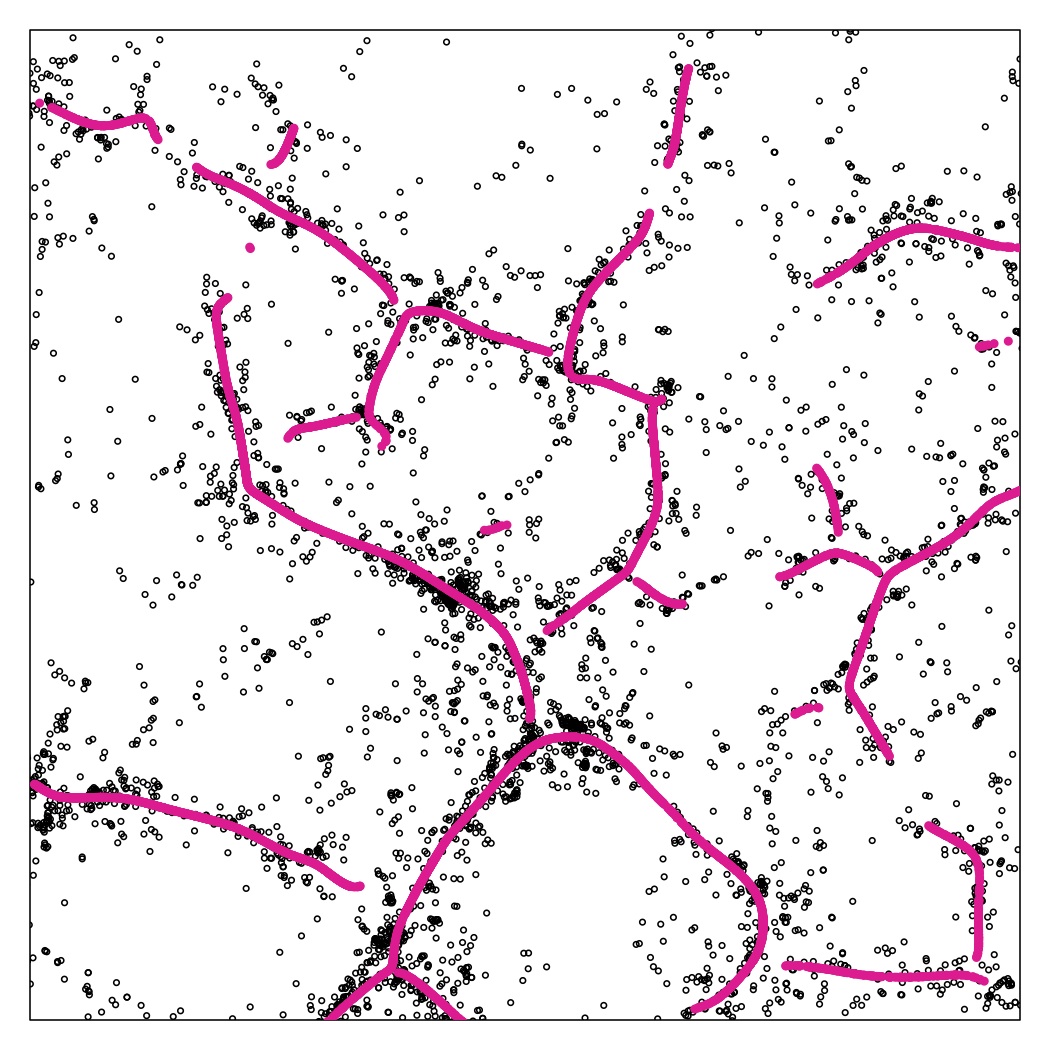}
\includegraphics[width=1.3in]{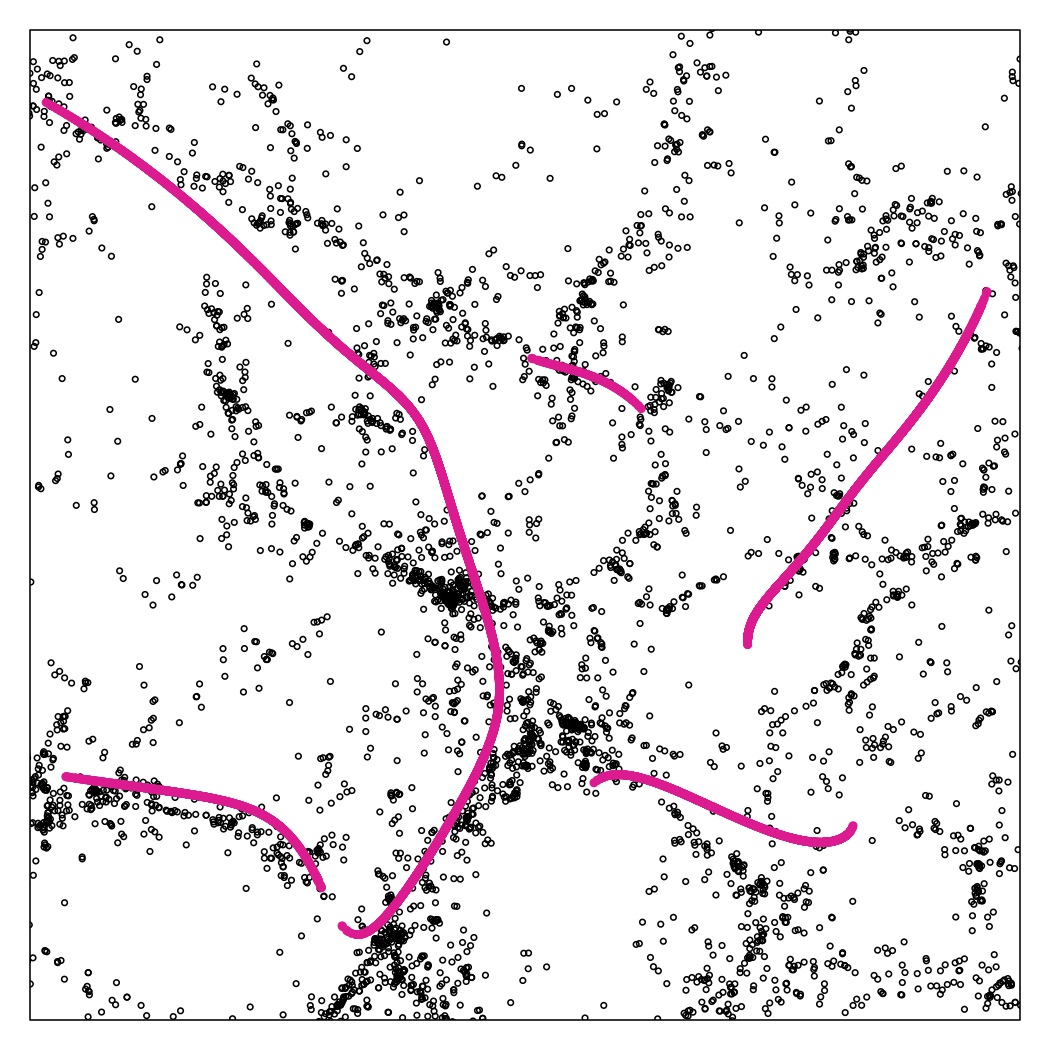}
\caption{
Another slice for the cosmic web data from the Sloan Digital Sky Survey.
The leftmost panel shows
the (estimated) $\cL_1$ coverage risk (right panel) for estimating density ridges 
under different smoothing parameters.
We estimated the $\cL_1$ coverage risk by using data splitting.
For the rest panels, from left to right, 
we display the case for under-smoothing, optimal smoothing, and over-smoothing.
As can be seen easily, the optimal smoothing method allows the SCMS
algorithm to detect the intricate cosmic network structure.}
\label{fig::cosmic2}
\end{figure}

\section{Discussion}
In this paper, we propose a method using coverage risk, a generalization
of mean integrated square error, to select the smoothing parameter for
the density ridge estimation problem.
We show that the coverage risk can be estimated using data splitting
or smoothed bootstrap and we derive the statistical consistency for risk estimators.
Both simulation and real data analysis show that the proposed 
bandwidth selector works very well in practice.

The concept of coverage risk is not limited to density ridges;
instead, it can be easily generalized to other manifold learning technique. 
Thus, we can use data splitting to estimate the risk and 
use the risk estimator to select the tuning parameters.
This is related to the so-called stability selection \cite{Rinaldo2010a},
which allows us to select tuning parameters even in an unsupervised learning settings.


\appendix
\section{Proofs}
Before we prove Theorem~\ref{thm::rate}, we need the following lemma
for comparing two curves.
\begin{lem}
Let $S_1,S_2$ be two bounded smooth curves in $\R^d$.
Let $\pi_{12}:S_1\mapsto S_2$ and $\pi_{21}:S_2\mapsto S_1$ be the projections
between them.
For $a\in S_1$ and $b\in S_2$, define $g_1(a)$ and $g_2(b)$ as the unit tangent vectors
for $S_1$ and $S_2$ at $a$ and $b$ respectively.
Assume $S_1$ and $S_2$ are similar in the following sense:
\begin{itemize}
\item[(S1)] $\pi_{12}$ and $\pi_{21}$ are one-one and onto,
\item[(S2)] the projections are similar:
$$
\max\left\{\sup_{x\in S_1}\norm{\pi_{12}(x)-\pi^{-1}_{21}(x)}, 
\sup_{x\in S_2}\norm{\pi_{21}(x)-\pi^{-1}_{12}(x)}\right\} = O(\epsilon_1),
$$
\item[(S3)] the tangent vectors are similar:
$$
\max\left\{\sup_{x\in S_1}|g_1(x)^Tg_2(\pi_{12}(x))|, 
\sup_{x\in S_2}|g_2(x)^Tg_1(\pi_{21}(x))|\right\} = 1+O(\epsilon_2),
$$
\item[(S4)] the length are similar:
$$
\length(S1) - \length(S2) = O(\epsilon_3)
$$
\end{itemize}
with $\epsilon_1,\epsilon_2, \epsilon_3$ being very small.
Let $\cI_1 = \int_{S_1} \norm{x-\pi_{12}(x)}^2 dx$
and $\cI_2 = \int_{S_2} \norm{y-\pi_{21}(y)}^2dy$.
Then we have
$$
|\cI_1 -\cI_2| = \sqrt{\cI_2} O(\epsilon_1) + \cI_2 O(\epsilon_2+\epsilon_3).
$$
Moreover, if we further assume 
\begin{itemize}
\item[(S5)] the Hausdorff distance $\Haus(S_1,S_2) = O(\epsilon_4)$ is small,
\end{itemize}
then for any function $\xi:\R^d \mapsto \R$ that has bounded continuous derivative,
we have
$$
\int_{0}^1\xi(\gamma_1(t))dt = \int_{0}^1\xi(\gamma_2(t))dt (1+ O(\epsilon_2+\epsilon_3+\epsilon_4)).
$$
\label{lem::curve}
\end{lem}
\begin{proof}
Since $S_1$ and $S_2$ are two bounded, smooth curves.
We may parametrized them by $\gamma_1:[0,1]\mapsto S_1$ and $\gamma_2: [0,1]\mapsto S_2$
with
\begin{equation}
\begin{aligned}
\gamma'_1(t) &= \tilde{g}_1(\gamma_1(t)), \gamma_1(0) = s_1,\\
\gamma'_2(t) &= \tilde{g}_2(\gamma_2(t)), \gamma_2(0) = s_2 = \pi_{12}(\gamma_1(0)),
\end{aligned}
\end{equation}
where $\tilde{g}_1= \ell_1 g_1$ and $\tilde{g}_2 = \ell_2 g_2$ for $\ell_1,\ell_2$
being the length of $S_1$ and $S_2$
and $s_1$ one of the end point of $S_1$.
The constant $\ell_j$ works as a normalization constant since $g_j$ is an unit vector;
it is easy to verify that
$$
\length(S_j) = \int_{0}^1 \norm{\tilde{g}_j(t)} dt = \int_0^1 \ell_j\norm{g_j(t)}dt
 = \ell_j.
$$
The starting point $s_2\in S_2$ must be the projection $\pi_{12}(s_1)$ otherwise the condition (S1) 
will not hold.

Let
\begin{equation}
\cI_1 = \int_{0}^1 \norm{\gamma_1(t)-\pi_{12}(\gamma_1(t))}^2 dt, \quad
\cI_2 = \int_{0}^1 \norm{\gamma_2(t)-\pi_{21}(\gamma_2(t))}^2 dt.
\end{equation}
Then the goal is to prove $\cI_1-\cI_2 = O(\epsilon_1^2) +O(\epsilon_2^2)$.

Now we consider another parametrization for $S_2$.
Let $\eta_2:[0,1]\mapsto S_2$ such that $\eta_2(t) = \pi_{12}(\gamma_1(t))$.
By (S1), $\eta_2$ is a parametrization for $S_2$.
The parametrization $\eta_2(t)$ has the following useful properties:
\begin{equation}
\begin{aligned}
\eta_2(0) &= \pi_{12}(\gamma_1(0)) = s_2, \\
\quad \eta_2'(t) &= g_2(\eta_2(t)) g_2(\eta_2(t))^T \gamma'_1(t)
 =g_2(\eta_2(t)) g_2(\pi_{12}(\gamma_1(t)))^T  \tilde{g}_1(\gamma_1(t)).
\end{aligned}
\end{equation}
By condition (S3) and (S4), we have
\begin{equation}
\begin{aligned}
g_2(\pi_{12}(\gamma_1(t)))^T  \tilde{g}_1(\gamma_1(t))
& = \ell_1 g_2(\pi_{12}(\gamma_1(t)))^T  g_1(\gamma_1(t))\\
& = \ell_1 (1+O(\epsilon_2))\\
& = \ell_2 (1+O(\epsilon_2)+O(\epsilon_3))
\end{aligned}
\end{equation}
uniformly for all $t\in [0,1]$.
Now apply this result to $\eta'_2(t)$, we obtain that
\begin{equation}
\quad \eta_2'(t) = g_2(\eta_2(t))(1+O(\epsilon_2)+O(\epsilon_3)).
\end{equation}
Together with $\eta_2(0) = \gamma_2(0)$, we have
\begin{equation}
\sup_{t\in[0,1]} \norm{\eta_2(t)-\gamma_2(t)} = O(\epsilon_2)+O(\epsilon_3).
\end{equation}

Now by definition of $\cI_1$ and the fact that $\pi^{-1}_{12}(\eta_2(t))=\gamma_1(t)$, 
we have
\begin{equation}
\begin{aligned}
\cI_1 &= \int_{0}^1 \norm{\gamma_1(t)-\pi_{12}(\gamma_1(t))}^2 dt\\
&= \int_{0}^1 \norm{\pi^{-1}_{12}(\eta_2(t))-\eta_2(t)}^2 dt\\
& = \int_0^1 \norm{\pi_{21}(\eta_2(t))+ O(\epsilon_1)  -\eta_2(t) }^2dt \quad \mbox{by (S2)}\\
& = \cI_2' + \sqrt{\cI_2'}O(\epsilon_1),
\label{eq::cI1}
\end{aligned}
\end{equation}
where $\cI_2' = \int_0^1 \norm{\pi_{21}(\eta_2(t))  
-\eta_2(t) }^2dt$.

Now we bound the difference between $\cI_2'$ and $\cI_2$.
Let $U$ be an uniform distribution over $[0,1]$
and define $h(x):[0,1]\mapsto \R$ as
$h(x) = \norm{\pi_{21}(\gamma_2(x))-\gamma_2(x)}$.
Note that it is easy to see that $h(x)$ has bounded derivative.
Then,
\begin{equation}
\cI_2 = \E\norm{\pi_{21}(\gamma_2(U))-\gamma_2(U)}^2 = \E h(U).
\end{equation}
Since both $\gamma_2$ and $\eta_2$ are parametrization for the 
curve $S_2$,
$\gamma^{-1}_2$ is well defined for all image of $\eta_2$.
We define the random variable $W = \gamma_2^{-1}(\eta_2(U))$.
Then by definition of $\cI_2'$,
\begin{equation}
\cI_2' = \E\norm{\pi_{21}(\eta_2(U))-\eta_2(U)}^2 = 
\E h(W).
\end{equation}
Since $\sup_{t\in[0,1]}\norm{\gamma_2'(t) - \eta_2'(t)} = O(\epsilon_2)+O(\epsilon_3)$,
we have $\gamma_2^{-1}(\eta_2(x)) = x+ O(\epsilon_2) + O(\epsilon_3)$.
Thus, 
the $p_W(t)-p_U(t) = O(\epsilon_2) + O(\epsilon_3)$,
where $p_W$ and $p_U$ are the probability density for random variable $W$ and $U$.
Since $U$ is uniform distribution, $p_U=1$ so that
\begin{equation}
\begin{aligned}
\E h(W) &= \int_{0}^1 h(t) p_W(t)dt\\
& = \int_{0}^1 h(t) (p_U(t)+O(\epsilon_2) + O(\epsilon_3))dt\\
& = \int_{0}^1 h(t)(1+O(\epsilon_2) + O(\epsilon_3))dt\\
& = \E h(U) (1+ O(\epsilon_2) + O(\epsilon_3)).
\end{aligned}
\label{eq::lem2::h}
\end{equation}
This
implies $\cI_2' = \cI_2 (1+ O(\epsilon_2) + O(\epsilon_3))$.
Therefore, by \eqref{eq::cI1} we conclude
\begin{equation}
\begin{aligned}
\cI_1 &=  \cI_2' + \sqrt{\cI_2'}O(\epsilon_1)\\
& = \cI_2 + \sqrt{\cI_2}O(\epsilon_1) + \cI_2(O(\epsilon_2) + O(\epsilon_3)),
\end{aligned}
\end{equation}
which completes the proof for the first assertion.

Now we prove the second assertion, here we will assume (S5).
Since $\xi$ has bounded first derivative, 
\begin{equation}
\begin{aligned}
\int_0^1 \xi(\gamma_1(t))dt
& = \int_0^1 \xi(\pi_{12}(\gamma_1(t)))dt (1+O(\Haus(S_1,S_2)))\\
& = \int_0^1 \xi(\eta_2(t))dt (1+O(\epsilon_4)).
\end{aligned}
\label{eq::lem::S5_01}
\end{equation}

Again, let $U$ be the uniform distribution and $W= \gamma_2^{-1}(\eta_2(U))$.
We now define the function $\tilde{h}(t) = \xi(\gamma_2(t))$ for $t\in [0,1]$.
Since both $\xi$ and $\gamma_2$ are bounded differentiable,
$\tilde{h} $ is also bounded differentiable.
Then it is easy to see that
\begin{equation}
\begin{aligned}
\int_0^1 \xi(\eta_2(t))dt &= \xi(\gamma_2(t) \gamma_2^{-1}(\eta_2(t)))dt = \E \tilde{h}(W)\\
\int_0^1 \xi(\gamma_2(t))dt & = \E \tilde{h}(U).
\label{eq::lem::S5_02}
\end{aligned}
\end{equation}
Now by the same derivation of \eqref{eq::lem2::h},
we conclude
\begin{equation}
\int_0^1 \xi(\eta_2(t))dt = \E \tilde{h}(W) = \E \tilde{h}(U) (1+O(\epsilon_2)+O(\epsilon_3)).
\label{eq::lem::S5_03}
\end{equation}
Thus, by \eqref{eq::lem::S5_01} and \eqref{eq::lem::S5_03}, we conclude
\begin{equation}
\int_0^1 \xi(\gamma_1(t))dt = \int_0^1 \xi(\gamma_2(t))dt(1+O(\epsilon_2)+O(\epsilon_3)+O(\epsilon_4)),
\end{equation}
which completes the proof.

\end{proof}

The following Lemma bounds the rate of convergence
for the kernel density estimator and will be used frequently in the following derivation.
\begin{lem}[Lemma 10 of \cite{chen2014asymptotic}; 
see also \cite{genovese2014nonparametric}]
Assume (K1--K2) 
and that
$\log n/n \leq h^d \leq b$ for some $0 < b < 1$.
Then we have 
\begin{align}
||\hat{p}_{n}-p||_{ k,\max} = O(h^2) + O_P\left(\sqrt{\frac{\log n}{nh^{d+2k}}}\right)
\end{align}
for $k=0,\cdots,3$.
Moreover,
\begin{equation}
\E||\hat{p}_{n}-p||_{ k,\max} = O(h^2) + O\left(\sqrt{\frac{\log n}{nh^{d+2k}}}\right).
\end{equation}
\label{lem::Krate}
\end{lem}

\begin{proof}[ for Theorem \ref{thm::rate}]
Here we prove the case for density ridges.
The case for density level set can be proved by the similar method.
We will use Lemma~\ref{lem::curve} to obtain the rate.
Our strategy is that first we derive
$\mathbb{E}(d(U_R,\hat{R}_n)^2)$
and then show that 
the other part $\mathbb{E}(d(U_{\hat{R}_n},R)^2)$
is similar to the first part.

{\bf Part 1.}
We first introduce the concept of reach \citep{Federer1959}.
For a smooth set $A$, the reach is defined as
\begin{equation}
\reach(A) = \inf\{r: \mbox{every point in $A\oplus r$ has an unique projection onto $A$.}\}.
\end{equation}
The reach condition is essential to establish a one-one projection between
two smooth sets.

By Lemma 2, property 7 of \cite{chen2014asymptotic}, 
\begin{equation}
\reach(R) \geq \min\left\{\frac{\delta_R}{2}, 
\frac{\beta_2^2}{A_2(\norm{p^{(3)}}_{\max}+\norm{p^{(4)}}_{\max})}\right\} 
\end{equation}
for some constant $A_2$.
Note that $\delta_R$ and $\beta_2$ are the constants in condition (R).

Thus, as long as $\hat{R}_n$ is close to $R$,
every point on $\hat{R}_n$ has an unique projection onto $R$.
Similarly, $\reach(\hat{R}_n)$ will have a similar bound to $\reach(R)$ whenever 
$\norm{\hat{p}_n-p}^*_{4,\max}$ is small (reach only depends on fourth derivatives).
Hence, every point on $R$ will have an unique projection onto $\hat{R}_n$.
The projections between $R$ and $\hat{R}_n$ will be one-one
and onto except for points near the end points for $R$ and $\hat{R}_n$.
That is,
when $\norm{\hat{p}_n-p}^*_{4,\max}$ is sufficiently small, 
there exists $R^\dagger\subset R$ and $\hat{R}_n^\dagger\subset \hat{R}_n$ 
such that the projection between $R^\dagger$ and $\hat{R}_n^\dagger$
are one-one and onto.
Moreover, the length difference
\begin{equation}
\begin{aligned}
\length(R)-\length(R^\dagger) &= O(\Haus(\hat{R}_n, R)),\\
\length(\hat{R}_n)-\length(\hat{R}_n^\dagger) &= O(\Haus(\hat{R}_n, R)).
\end{aligned}
\label{eq::length}
\end{equation}
Note that by Theorem 6 in \cite{genovese2014nonparametric},
\begin{equation}
\Haus(\hat{R}_n, R) = O(\norm{\hat{p}_n-p}^*_{2,\max}).
\label{eq::Haus}
\end{equation}

Let $x\in R^\dagger$,
and let $x' =\pi_{\hat{R}_n}(x)\in \hat{R}^\dagger_n$ be its projection onto $\hat{R}_n$.
Then by Theorem 3 in \cite{chen2014asymptotic}
(see their derivation in the proof, the empirical approximation,
page 30-32 and equation (79)),
we have
\begin{equation}
x'-x = W_2(x) (\hat{g}_n(x)-g(x))(1+O(\norm{\hat{p}_n-p}^*_{3,\max})),
\end{equation}
where 
\begin{equation}
\begin{aligned}
W_2(x) &= N(x)H_N^{-1}(x)N(x)\\
H_N(x) & = N(x)^T H(x)N(x)
\end{aligned}
\end{equation}
and $N(x)$ is a $d\times (d-1)$ matrix called the \emph{normal matrix} for $R$ at $x$ 
whose columns space spanned the normal space for $R$ at $x$.
The existence for $N(x)$ is given in Section 3.2 and Lemma 2 in \cite{chen2014asymptotic}.
Thus, we have 
\begin{equation}
\E\left(d(x, \hat{R}_n)^2\right) = \E\left(\norm{x-x'}^2\right) 
= \E\left\|W_2(x) (\hat{g}_n(x)-g(x))\right\|^2 +\Delta_n,
\end{equation}
where $\Delta_n$ is the remaining term
and by Cauchy-Schwartz inequality, 
$$
\Delta_n  \leq \E\left\|W_2(x) (\hat{g}_n(x)-g(x))\right\|^2 O(\E\norm{\hat{p}_n-p}^{*}_{3,\max}).
$$
Thus, 
\begin{equation}
\begin{aligned}
\E\left(d(x, \hat{R}_n)^2\right) 
&= \E\left\|W_2(x) (\hat{g}_n(x)-g(x))\right\|^2 +\Delta_n\\
& = \E\left\|W_2(x) (\hat{g}_n(x)-\mathbb{E}(\hat{g}_n(x))+\mathbb{E}(\hat{g}_n(x))-g(x))\right\|^2+\Delta_n\\
& = \Tr(\Cov(W_2(x)\hat{g}_n(x)) )+ \norm{W_2(x) (\mathbb{E}(\hat{g}_n(x))-g(x))}^2 +\Delta_n\\
& = \frac{1}{nh^{d+2}}\Tr(\Sigma(x)) + h^4 b(x)^Tb(x)
+ o\left(\frac{1}{nh^{d+2}}\right) + o\left(h^4\right),
\label{eq::main}
\end{aligned}
\end{equation}
where
\begin{equation}
\begin{aligned}
\Sigma(x) &= W_2(x)\Sigma(K) W_2(x)p(x),\\
b(x) & = c(K)W_2(x)\nabla (\nabla^2 p(x))
\end{aligned}
\label{eq::parameter}
\end{equation}
are related to the variance and bias for nonparametric gradient estimation
($\Sigma(K)p(x)$ is the asymptotic covariance matrix for $\hat{p}_n$ and
$c(K) \nabla (\nabla^2 p(x))$ is the asymptotic bias for $\hat{p}_n$).
$\Sigma(K)$ is a matrix and $c(K)$ is a scalar; they both depends only on the 
kernel function $K$.
$\nabla^2 = \frac{\partial^2}{\partial x_1^2} +\cdots+ \frac{\partial^2}{\partial x_d^2}$
is the Laplacian operator.

Now we compute $\E(d(U_R, \hat{R}_n)^2)$.
Note that since the length difference between $R$ and 
$R^\dagger$ is bounded by \eqref{eq::length} and \eqref{eq::Haus}:
\begin{equation}
\begin{aligned}
\mathbb{P}(U_R\in R^\dagger)  &= 1-O(\mathbb{E}(\norm{\hat{p}_n-p}^*_{2,\max}))\\
& =1- O(h^2) -O\left(\sqrt{\frac{\log n}{nh^{d+4}}}\right).
\end{aligned}
\label{eq::Lprob}
\end{equation}
Note that we use Lemma~\ref{lem::Krate} to convert the norm into probability bound.
By tower property (law of total expectation),
\begin{equation}
\begin{aligned}
\E(d(U_R, \hat{R}_n)^2) &= \E(\E(d(U_R, \hat{R}_n)^2|U_R))\\
&=\E(\E(d(U_R, \hat{R}_n)^2|U_R, U_R\in R^\dagger))\mathbb{P}(U_R\in R^\dagger)\\
&\quad+ \E(\E(d(U_R, \hat{R}_n)^2|U_R, U_R\notin R^\dagger))\mathbb{P}(U_R\notin R^\dagger)\\
&=\E\left(\frac{1}{nh^{d+2}}\Tr(\Sigma(U_R)) + h^4 b(U_R)^Tb(U_R)\right)
+ o\left(\frac{1}{nh^{d+2}}\right) + o\left(h^4\right).
\end{aligned}
\end{equation}
Note that by \eqref{eq::Lprob}, the contribution from $\mathbb{P}(U_R\notin R^\dagger)$
is smaller than the main effect in \eqref{eq::main} so we absorb it into the small $o$ terms.
Defining $B_R^2  = \E( b(U_R)^Tb(U_R))$ and $\sigma_R^2 = \E(\Tr(\Sigma(U_R)))$,
we obtain 
\begin{equation}
\E(d(U_R, \hat{R}_n)^2) = B_R^2 h^4 + \frac{\sigma_R^2}{nh^{d+2}} 
+ o\left(\frac{1}{nh^{d+2}}\right) + o\left(h^4\right).
\label{eq::lem2::p1}
\end{equation}

{\bf Part 2.}
We have proved the first part for the $\cL_2$ coverage risk.
Now we prove the result for $\E(d(U_{\hat{R}_n}, R)^2)$; 
this will apply Lemma~\ref{lem::curve}.
If we think of $R^\dagger$ as $S_1$ and $\hat{R}_n^\dagger$ as 
$S_2$ in Lemma~\ref{lem::curve},
then
\begin{equation}
\begin{aligned}
\E(d(U_{R^\dagger}, \hat{R}_n^\dagger)^2|X_1,\cdots,X_n) 
=  \int_{0}^1 \norm{\gamma_1(t)-\pi_{12}(\gamma_1(t))}^2 dt = \cI_1\\
\E(d(U_{\hat{R}_n^\dagger}, R^\dagger)^2|X_1,\cdots,X_n) =  
\int_{0}^1 \norm{\gamma_2(t)-\pi_{21}(\gamma_2(t))}^2 dt = \cI_2.
\end{aligned}
\end{equation}
Thus, $\E(d(U_{\hat{R}_n^\dagger}, R^\dagger)^2)$
is approximated by $\E(d(U_{R^\dagger}, \hat{R}_n^\dagger)^2)$
if the $\epsilon_1,\epsilon_2,\epsilon_3$ in Lemma~\ref{lem::curve} is small.
Here we bound $\epsilon_j$.

The bound for $\epsilon_1$ is simple. For all $x\in S_1$,  
let $\theta$ be the angle between the two vectors $v_1 = \pi_{12}(x)-x$
and $v_2 =\pi_{21}^{-1}(x)-x$.
By the property of projection, $v_1$ is normal to $\hat{R}_n$ at $\pi_{12}(x)$
and $v_2$ is normal to $R$ at $x$.
Thus, by Lemma 2 properties 5 and 6 of \cite{chen2014asymptotic}, 
the angle $\theta$ is bounded by $O(\norm{\hat{p}_n-p}^{*}_{3,\max})$.
Note that their Lemma proves the normal matrices $N(x)$
and $\hat{N}_n(\pi_{12}(x))$ are close which implies
the canonical angle between two subspace are close
so that $\theta$ is bounded.
Now by the fact that both $\norm{\pi_{12}(x)-x}$ and $\norm{\pi^{-1}_{21}(x)-x}$
are bounded by $\Haus(\hat{R}_n, R)$,
we conclude $\epsilon_1 \leq \Haus(\hat{R}_n, R)\times\theta
= O(\norm{\hat{p}_n-p}^{*2}_{3,\max})$.

For $\epsilon_2$, we will use the property of normal matrix $N(x)$.
Let $\hat{N}_n(x)$ be the normal matrix for $\hat{R}_n$ at $x$.
By Lemma 2, properties 5 and 6 of \cite{chen2014asymptotic},
\begin{align*}
\norm{N(x)N(x)^T - \hat{N}_n(\pi_{\hat{R}_n}(x))\hat{N}_n(\pi_{\hat{R}_n}(x))^T}_{\max}
 &= O (\Haus(\hat{R}_n, R)) + O(\norm{\hat{p}_n-p}^*_{3,\max}) \\
 &= O(\norm{\hat{p}_n-p}^*_{3,\max}).
\end{align*}
$N(x)N(x)^T$ is the projection matrix onto normal space;
so the tangent vector is perpendicular to that projection.
The bounds for the two projection matrix implies the bound
to the two tangent vectors.
Thus,  $\epsilon_2 = O(\norm{\hat{p}_n-p}^*_{3,\max})$.

For $\epsilon_3$, since the smoothness for $\hat{R}_n$
is similar to $R$ (the normal direction is similar by $\epsilon_2$)
and their Hausdorff distance is bounded by $O(\norm{\hat{p}_n-p}^*_{2,\max})$.
The length difference is at the same rate of Hausdorff distance.
Thus, we may pick $\epsilon_3 = O(\norm{\hat{p}_n-p}^*_{2,\max})$.

Let $\cI_1 = \E(d(U_{R^\dagger}, \hat{R}_n^\dagger)^2|X_1,\cdots,X_n)$
and $\cI_2 = \E(d(U_{\hat{R}_n^\dagger}, R^\dagger)^2|X_1,\cdots,X_n)$.
By Lemma~\ref{lem::curve} and the above choice for $\epsilon_j$,
we conclude
\begin{equation}
\begin{aligned}
\cI_1
&= \cI_2(1+ O(\norm{\hat{p}_n-p}^{*}_{3,\max})) + \sqrt{\cI_2}O(\norm{\hat{p}_n-p}^{*2}_{3,\max}).
\end{aligned}
\end{equation}
Thus, by tower property again (taking expectation over both side)
and Lemma \ref{lem::Krate}
$\E \norm{\hat{p}_n-p}^*_{3,\max} = O(h^2) + O\left(\sqrt{\frac{\log n}{nh^{d+6}}}\right) = o(1)$,
\begin{equation}
\E(d(U_{R^\dagger}, \hat{R}_n^\dagger)^2) = 
\E(\cI_1) =\E(\cI_2) + o(1) = 
\E(d(U_{\hat{R}_n^\dagger}, R^\dagger)^2) +o(1).
\label{eq::lem2::p2::1}
\end{equation}
Now since by \eqref{eq::length} and the fact that
$\E \Haus(\hat{R}_n, R) = o(1)$, we have 
\begin{equation}
\begin{aligned}
\E(d(U_{R^\dagger}, \hat{R}_n^\dagger)^2) &= \E(d(U_{R}, \hat{R}_n)^2)(1+ o(1))\\
\E(d(U_{\hat{R}_n^\dagger}, R^\dagger)^2) &= \E(d(U_{\hat{R}_n}, R)^2)(1+ o(1)).
\label{eq::lem2::p2::2}
\end{aligned}
\end{equation}

Combining  by \eqref{eq::lem2::p1}, \eqref{eq::lem2::p2::1} and \eqref{eq::lem2::p2::2}, 
we conclude
\begin{equation}
\begin{aligned}
\Risk_{2,n} &= \frac{\E(d(U_{R}, \hat{R}_n)^2)+ \E(d(U_{\hat{R}_n}, R)^2)}{2}\\
&= \E(d(U_{R}, \hat{R}_n)^2)+o(1)\\ 
&= 
B_R^2 h^4 + \frac{\sigma_R^2}{nh^{d+2}} 
+ o\left(\frac{1}{nh^{d+2}}\right) + o\left(h^4\right),
\end{aligned}
\end{equation}
where 
$B_R^2  = \E( b(U_R)^Tb(U_R))$ and $\sigma_R^2 = \E(\Tr(\Sigma(U_R)))$.
Note that all the above derivation works only when 
\begin{equation}
\E \norm{\hat{p}_n-p}^*_{3,\max} = O(h^2) + O\left(\sqrt{\frac{\log n}{nh^{d+6}}}\right) = o(1).
\end{equation}
This requires $h\rightarrow 0$ and $\frac{\log n}{nh^{d+6}}\rightarrow 0$, which
constitutes the conditions on $h$ we need.
 
\end{proof}

\begin{proof}[ for Theorem~\ref{thm::estimate}]
Since we are proving the bootstrap consistency,
we assume $X_1,\cdots,X_n$ are given.

By Theorem~\ref{thm::rate}, 
the estimated risk $\hat{\Risk}_{n,2}$ 
has the following asymptotic behavior
\begin{equation}
\hat{\Risk}_{n,2} =\hat{B}_R^2 h^4 + \frac{\hat{\sigma}_R^2}{nh^{d+2}} 
+ o\left(\frac{1}{nh^{d+2}}\right) + o\left(h^4\right),
\label{eq::thm2::eq0}
\end{equation}
where
\begin{equation}
\begin{aligned}
\hat{B}_R^2 & =  \E\left( \hat{b}_n(U_{\hat{R}_n})^T\hat{b}_n(U_{\hat{R}_n})|X_1,\cdots,X_n\right),\\
\hat{\sigma}_R^2 &= \E\left(\Tr(\hat{\Sigma}_n(U_{\hat{R}_n}))|X_1,\cdots,X_n\right)
\end{aligned}
\label{eq::thm2::eq1}
\end{equation}
with $\hat{b}_n(x) = c(K)W_2(x)\nabla (\nabla^2 \hat{p}_n(x))$
and $\hat{\Sigma}_n(x) = W_2(x)\Sigma(K) W_2(x)\hat{p}_n(x)$
from \eqref{eq::parameter}.
To prove the bootstrap consistency,
it is equivalent to prove that $\hat{B}_R^2$
and $\hat{\sigma}^2_R$ converges to $B_R$ and $\sigma_R^2$.

Here we prove the consistency for $\hat{B}_R$.
The consistency for $\hat{\sigma}_R$ can be proved in the similar way.
We define the following two functions
\begin{equation}
\begin{aligned}
\hat{\Omega}_n (x) &= \norm{c(K)W_2(x)\nabla (\nabla^2 \hat{p}_n(x))}^2,\\
\Omega (x) &= \norm{c(K)W_2(x)\nabla (\nabla^2 p(x))}^2.
\end{aligned}
\label{eq::thm2::eq2}
\end{equation}
It is easy to see that $\hat{B}_R^2 = \E\left(\hat{\Omega}_n(U_{\hat{R}_n})|X_1,\cdots,X_n\right)$
and $B_R^2 = \E\left(\Omega(U_R)\right)$.

Similarly as in the proof for Theorem~\ref{thm::rate},
we define $\hat{R}_n^\dagger\subset\hat{R}_n$ that has
one-one and onto projection to $R^\dagger$.
By \eqref{eq::Lprob}, we can replace $U_{\hat{R}_n}$ by $U_{\hat{R}_n^\dagger}$
and $U_R$ by $U_{R^\dagger}$ at the cost of probability 
$O(h^2) + O\left(\sqrt{\frac{\log n}{nh^{d+4}}}\right)$.

Now we will apply Lemma~\ref{lem::curve} again to prove the result.
Again, we think of $R^\dagger$ as $S_1$ and $\hat{R}_n^\dagger$ as $S_2$.
Let $U$ be an uniform distribution over $[0,1]$.
Then the random variable $U_{R^\dagger} = \gamma_1(U)$ 
and $U_{\hat{R}_n^\dagger} = \gamma_2(U)$.
Thus,
\begin{equation}
\E\left(\Omega(U_{R^\dagger})\right) = \int_0^1 \Omega(\gamma_1(t))dt,\quad
\E\left(\hat{\Omega}_n(U_{\hat{R}_n^\dagger})|X_1,\cdots,X_n\right) = \int_0^1 \hat{\Omega}_n(\gamma_2(t))dt.
\label{eq::thm2::eq3}
\end{equation}
By the second assertion in Lemma~\ref{lem::curve}, 
\begin{equation}
\begin{aligned}
\E\left(\Omega(U_{R^\dagger})\right) & = \int_0^1 \Omega(\gamma_1(t))dt\\
& = \int_0^1 \Omega(\gamma_2(t))dt(1+O(\epsilon_2)+O(\epsilon_3)+O(\epsilon_4))\\
& = \int_0^1 \Omega(\gamma_2(t))dt (1+O(\norm{\hat{p}_n-p}^*_{3,\max})).
\end{aligned}
\label{eq::thm2::eq4}
\end{equation}
Note that we use the fact that $\Haus(\hat{R}_n,R) = O(\norm{\hat{p}_n-p}^*_{2,\max})$.
Since $\Omega$ only involves third derivative for the density $p$,
we have $\sup_{x\in \R^d}\norm{\Omega(x)-\hat{\Omega}_n(x)} = O(\norm{\hat{p}_n-p}_{3,\max})$.
This implies 
\begin{equation}
 \int_0^1 \Omega(\gamma_2(t))dt  
 =  \int_0^1 \hat{\Omega}_n(\gamma_2(t))dt  + O(\norm{\hat{p}_n-p}_{3,\max}).
 \label{eq::thm2::eq5}
\end{equation}

Now combining all the above and the definition for $\hat{B}_R$,
we conclude
\begin{equation}
\begin{aligned}
\hat{B}_R^2 &= \E\left(\hat{\Omega}_n(U_{\hat{R}_n})|X_1,\cdots,X_n\right)\\
& = \E\left(\hat{\Omega}_n(U_{\hat{R}_n^\dagger})|X_1,\cdots,X_n\right) + O(\Haus(\hat{R}_n,R))\\
& =  \int_0^1 \hat{\Omega}_n(\gamma_2(t))dt + O(\Haus(\hat{R}_n,R)) \quad \mbox{ (by \eqref{eq::thm2::eq3})}\\
& =  \int_0^1 \Omega(\gamma_2(t))dt  + O(\norm{\hat{p}_n-p}_{3,\max})\quad \mbox{ (by \eqref{eq::thm2::eq5})}\\
& = \E\left(\Omega(U_{R^\dagger})\right)+ O(\norm{\hat{p}_n-p}_{3,\max})\quad \mbox{ (by \eqref{eq::thm2::eq4})}\\
& = \E\left(\Omega(U_R)\right)+ O(\norm{\hat{p}_n-p}_{3,\max})\\
& = B_R^2+ O(\norm{\hat{p}_n-p}_{3,\max}).
\end{aligned}
\end{equation}

Therefore, as along as we have $\norm{\hat{p}_n-p}_{3,\max} = o_P(1)$,
we have 
\begin{equation}
\hat{B}_R^2-B_R^2 = o_P(1).
\label{eq::thm2::eq6}
\end{equation}
Similarly, we the same condition implies
\begin{equation}
\hat{\sigma}_R^2-\sigma_R^2 = o_P(1).
\label{eq::thm2::eq7}
\end{equation}
Now recall from \eqref{eq::thm2::eq0} and Theorem~\ref{thm::rate},
the risk difference is
\begin{equation}
\begin{aligned}
\hat{\Risk}_{n,2}-\Risk_{n,2} &= (\hat{B}_R^2-B_R^2)h^4 + \frac{\hat{\sigma}_R^2-\sigma_R^2}{nh^{d+2}}
+o\left(h^4\right)+o\left(\frac{1}{nh^{d+2}}\right)\\
& = o_P\left(h^4\right)+o_P\left(\frac{1}{nh^{d+2}}\right) \quad 
\mbox{(by \eqref{eq::thm2::eq6} and \eqref{eq::thm2::eq7})}.
\end{aligned}
\label{eq::thm2::risk2}
\end{equation}

Since Theorem~\ref{thm::rate} implies 
$\Risk_{n,2} = O\left(h^4\right)+O\left(\frac{1}{nh^{d+2}}\right)$,
by \eqref{eq::thm2::risk2} we have
\begin{equation}
\frac{\hat{\Risk}_{n,2}-\Risk_{n,2}}{\Risk_{n,2}} = o_P(1)
\label{eq::thm2::risk3}
\end{equation}
which proves the theorem.

Note that in order \eqref{eq::thm2::risk3} to hold, we need
$\norm{\hat{p}_n-p}_{3,\max} = o_P(1)$.
By Lemma~\ref{lem::Krate},
\begin{equation}
\norm{\hat{p}_n-p}_{3,\max} = O(h^2) + O_P\left(\sqrt{\frac{\log n}{nh^{d+6}}}\right).
\end{equation}
Thus, a sufficient condition to $\norm{\hat{p}_n-p}_{3,\max} = o_P(1)$
is to pick $h$ such that $\frac{\log n}{nh^{d+6}}\rightarrow 0$ and $h\rightarrow 0$.
This gives the restriction for the smoothing parameter $h$.

\end{proof}

\newpage
\small

\bibliographystyle{abbrvnat}
\bibliography{CD.bib}

---
\end{document}